\documentclass[
    a4paper,
    oneside,
    5p,
    nonatbib,
    sort&compress,
    times]{elsarticle}

\makeatletter
\let\c@author\relax
\makeatother

\makeatletter
\let\bibfont\relax
\makeatother

\usepackage[backend=bibtex, sorting=nyt, firstinits=true, url=false]{biblatex}
\usepackage{amsmath}
\usepackage{amsthm}
\usepackage{multicol}
\usepackage{amsfonts}
\usepackage{amssymb}
\usepackage{bm}
\usepackage{enumitem}
\usepackage{tikz}
\usepackage{pgfplots}
\usepackage[labelfont=bf]{caption}
\usepackage{subcaption}
\usepackage{booktabs}
\usepackage[
    hidelinks,
    pdfproducer={Latex with hyperref},
    pdfcreator={pdflatex}]{hyperref}
\usepackage{theoremref}
\usepackage[acronym]{glossaries}
\usepackage{glossary-inline}
\usepackage{graphicx}

\usepgfplotslibrary{groupplots}
\usepgfplotslibrary{external}
\setglossarystyle{inline}
\renewcommand*{\glossarysection}[2][]{}

\hypersetup{bookmarksnumbered}

\theoremstyle{plain}
\newtheorem{proposition}{Proposition}
\newtheorem{corollary}{Corollary}

\theoremstyle{definition}
\newtheorem{remark}{Remark}
\newtheorem{definition}{Definition}
\newtheorem{example}{Example}
\newtheorem{algorithm}{Algorithm}

\newcommand{\figurewidth}{0.5\columnwidth}
\newcommand{\subfigurewidth}{0.25\textwidth}

\renewcommand{\figurename}{Fig.}

\newcommand{\orcid}[1]{
  \href{https://orcid.org/#1}{
      \includegraphics[height=2ex]{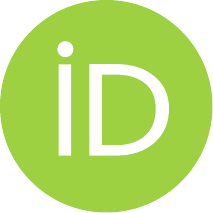}
      \url{https://orcid.org/#1}
  }
}

\newcommand{\abbreviations}[1]{%
  \nonumnote{\textit{Abbreviations:\enspace}#1}}

\makeglossaries

\newacronym[
    plural=pdfs,
    firstplural=probability density functions]{pdf}{pdf}{
        probability density function
}

\newacronym[
    plural=CDFs,
    firstplural=cumulative distribution functions]{cdf}{CDF}{
        cumulative distribution function
}

\newacronym{mle}{MLE}{maximum likelihood estimation}

\newacronym{pll}{PLL}{penalized log-likelihood}

\newacronym[
    plural=rvs,
    firstplural=random variables]{rv}{rv}{random variable}

\newacronym{iff}{iff}{if and only if}

\newacronym[
    plural=EVCs,
    firstplural=extreme-value copulas]{evc}{EVC}{extreme-value copula}

\newacronym[
    plural=PFs,
    firstplural=Pickands functions]{pf}{PF}{Pickands function}

\newacronym[
    plural=SBEVCs,
    firstplural=semiparametric bivariate extreme-value copulas]{sbevc}{SBEVC}{\textit{semiparametric bivariate extreme-value copula}}

\newacronym{cobs}{COBS}{\textit{constrained B-splines}}

\newacronym[
    plural=WTs,
    firstplural=Williamson transforms]{wt}{WT}{Williamson transform}

\newacronym{ss}{SS}{simulation study}

\newacronym{gc}{GC}{Gini coefficient}

\newacronym{clr}{CLR}{centred log-ratio}

\newacronym{tvd}{TVD}{total variation distance}

\newacronym{rmise}{RMISE}{root mean integrated squared error}

\newacronym[
    plural=GWs,
    firstplural=gravitational waves]{gw}{GW}{gravitational wave}

\newacronym[
    plural=BMs,
    firstplural=binary mergers]{bm}{BM}{binary merger}

\newacronym[
    plural=ZBSs,
    firstplural=zero-integral B-splines]{zbs}{ZBS}{zero-integral B-spline}

\abbreviations{\printglossary[type=\acronymtype]}

\begin{document}
    
\newcommand\papertitle{Semiparametric bivariate extreme-value copulas}

\title{\papertitle\tnoteref{funding}}

\tnotetext[funding]{
    This research did not receive any specific grant from funding agencies
    in the public, commercial, or not-for-profit sectors.
    Declarations of interest: none.
}

\newcommand\javierfdezserrano{Javier Fern\'andez Serrano}
\newcommand\emailjavierfdezserrano{javier.fernandezs01@estudiante.uam.es}
\newcommand\orcidjavierfdezserrano{0000-0001-5270-9941}

\author{\javierfdezserrano\fnref{orcid}}
\address{
    Departamento de Matem\'aticas, Universidad Aut\'onoma de Madrid, Madrid, Spain
}
\ead{\emailjavierfdezserrano}
\fntext[orcid]{\orcid{\orcidjavierfdezserrano}}

\date{}

\hypersetup{
  pdftitle={\papertitle},
  pdfauthor={\javierfdezserrano}
}

    \newcommand\abstractmeta {
    Extreme-value copulas arise as the limiting dependence structure of component-wise maxima.
    Defined in terms of a functional parameter, they are one of the most widespread copula families due to their flexibility and ability to capture asymmetry.
    Despite this, meeting the complex analytical properties of this parameter in an unconstrained setting remains a challenge, restricting most uses to models with very few parameters or nonparametric models.
    Focusing on the bivariate case, we propose a novel semiparametric approach.
    Our procedure relies on a series of transformations, including Williamson's transform and starting from a zero-integral spline.
    Without further constraints, wholly compliant solutions can be efficiently obtained through maximum likelihood estimation, leveraging gradient optimization.
    We successfully conducted several experiments on simulated and real-world data.
    Our method outperforms another well-known nonparametric technique over small and medium-sized samples in various settings.
    Its expressiveness is illustrated with precious data gathered by the gravitational wave detection LIGO and Virgo collaborations.
}

\begin{abstract}
    \abstractmeta
\end{abstract}

\hypersetup {
  pdfsubject={\abstractmeta},
}

\newcommand\extremevaluecopula{extreme-value copula}
\newcommand\bivariatecopula{bivariate copula}
\newcommand\semiparametric{semiparametric model}
\newcommand\williamson{Williamson's transform}
\newcommand\compositionalspline{compositional spline}

\newcommand\msccopulas{62H05}

\newcommand\mscmultivariate{62H12}

\newcommand\msccomputationalmethods{62-08}

\begin{keyword}
    \bivariatecopula \sep
    \compositionalspline \sep
    \extremevaluecopula \sep
    \semiparametric \sep
    \williamson
    \newline
    \MSC[2020]{
      Primary
      \msccopulas\sep
      \mscmultivariate\sep
      Secondary
      \msccomputationalmethods
    }
\end{keyword}

\hypersetup {
  pdfkeywords={
    \bivariatecopula,
    \compositionalspline,
    \extremevaluecopula,
    \semiparametric,
    \williamson
  },
}

    \maketitle

    \section{Introduction}

\label{sec:introduction}

A copula $C$ is an \gls{evc} if it is the weak limit of copulas emerging from component-wise maxima~\cite{Gudendorf2010}.
In the bivariate case, \glspl{evc} can be expressed as
\begin{equation*}
    C(u, v) = \exp \left\{
        \log (uv) A \left[ \frac{\log(u)}{\log(uv)} \right]
    \right\}\,,
    \
    \text{for} \ u, v \in  (0, 1)^2 \,,
\end{equation*}
where $A: [0, 1] \longrightarrow \mathbb{R}$, known as the \gls{pf}, satisfies the following two constraints:
\begin{enumerate}[noitemsep]
    \item
        $\max\{t, 1 - t\} \leq A(t) \leq 1$, for all $t \in [0, 1]$.
    \item
        $A$ is convex.
\end{enumerate}
The segments making the lower bound for $A$ are called the \textit{support} lines of the \gls{pf}.
\figurename~\ref{fig:pickands} shows the \gls{pf} geometry.

\begin{figure}[ht]
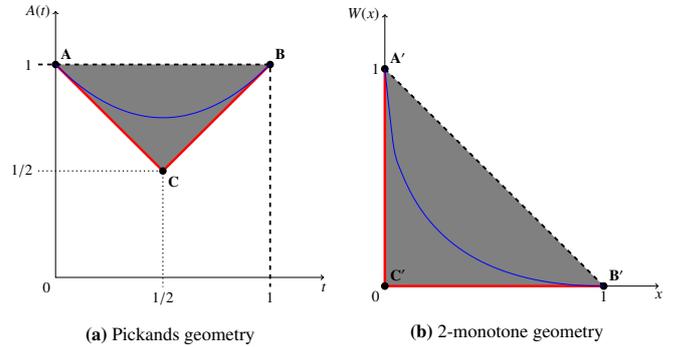

    \centering
    \begin{subfigure}{\figurewidth}
        \centering
        \resizebox{\columnwidth}{!}{
            \input{tikz/01-a-pickands-geometry.tex}
        }
        \caption{Pickands geometry}
        \label{fig:pickands}
    \end{subfigure}%
    \begin{subfigure}{\figurewidth}
        \centering
        \resizebox{\columnwidth}{!}{
            \input{tikz/01-b-tawn-pickands.tex}
        }
        \caption{2-monotone geometry}
        \label{fig:williamson}
    \end{subfigure}
    \caption{
        On the left, \gls{pf} geometry.
        The admissible region for its graph appears in grey.
        The support lines show in red.
        An example of \gls{pf}, namely $A(t) = t^2 - t + 1$, is drawn in blue.
        On the right, the geometry of a 2-monotone function derived from a \gls{pf} through an affine transformation mapping $\mathbf{A}$, $\mathbf{B}$ and $\mathbf{C}$ to $\mathbf{A}'$, $\mathbf{B}'$ and $\mathbf{C}'$.
        The transformed version of the \gls{pf} on the left, $W(x) = x - 2 \sqrt{x} + 1$, is drawn in blue.
        The graph of the function $W$ ranges from $\mathbf{A}'$ to $\mathbf{B}'$ and never crosses the dashed line segment between these two points.
    }
    \label{fig:pickands-and-williamson}
\end{figure}

The \gls{pf} constraints are not inherently satisfied by conventional approximation methods.
Thus, most \gls{evc} modelling depends on well-known one-parameter symmetrical families, like in Table~\ref{tab:evc-families}.
\citeauthor{Khoudraji1995}'s device allows obtaining asymmetrical \glspl{evc} from the latter, somewhat extending the applicability of parametric models~\cite{Khoudraji1995, Eschenburg2013}.
If $A$ is a \gls{pf}, given $\alpha, \beta \in (0, 1]$, the following will also be a \gls{pf}:
\begin{equation}
    \label{eq:pickands-khoudraji}
    A_{\alpha, \beta}(t) = (1 - t)(1 - \alpha) + t (1 - \beta) +
    w \ A \left( \frac{t \beta}{w} \right) \,,
\end{equation}
where $w = (1 - t) \alpha + t \beta$.
Nonparametric methods are currently the only alternative.

\begin{table}
    \centering
    \resizebox{0.7\columnwidth}{!} {
        \begin{tabular}{lcr}
            \toprule
            Family & $A_{\theta}(t)$ & $\theta$ range \\
            \midrule
            Gumbel &
            $\left[ t^{\theta} + (1 - t)^{\theta} \right]^{1/\theta}$ &
            $[1, \infty)$ \\
            Galambos &
            $1 - \left[ t^{-\theta} + (1 - t)^{-\theta} \right]^{-1/\theta}$ &
            $(0, \infty)$ \\
            \bottomrule
        \end{tabular}
    }
    \caption{
        Main one-parameter \gls{evc} families~\cite{Kamnitui2019}.
        The above families are symmetrical.
        \citeauthor{Khoudraji1995}'s procedure provides a means for introducing asymmetry.
    }
    \label{tab:evc-families}
\end{table}

\paragraph{Related work}

\citeauthor{Vettori2017} provide a comprehensive review of \glspl{evc}~\cite{Vettori2017}.
In general, nonparametric models demonstrate greater flexibility than parametric ones.
Parametric models using~\citeauthor{Khoudraji1995}'s device perform well against nonparametric ones in dimensions higher than two and with mild asymmetry.

One of the first estimators for the \gls{pf} was proposed by~\citeauthor{Pickands1981} in bivariate survival analysis~\cite{Pickands1981}.
However,~\citeauthor{Pickands1981}'s method produces almost surely~\cite{Jimenez2001} non-convex \glspl{pf} over $[0, 1]$.
\citeauthor{Pickands1981} himself proposed in~\cite{Pickands1981} to use the greatest convex minorant of the original estimator, which remains one of the most practical and efficient approaches.

Perhaps the most widespread nonparametric method is due to~\citeauthor{Caperaa1997}~\cite{Caperaa1997}, from which it borrows its name CFG.
They observe that, given a random sample $\{(U_i, V_i)\}_{i = 1}^n$ from an \gls{evc} with \gls{pf} $A$, the transformation $Z_i = \log U_i / \log (U_i V_i)$ is distributed according to the \gls{cdf}
\begin{equation}
    \label{eq:h-cdf}
    H(z) = z + z(1 - z) \frac{A'(z)}{A(z)} \,.
\end{equation}
One can empirically estimate $H$ with some $\tilde{H}$ and solve~\eqref{eq:h-cdf} for an estimator
\begin{equation}
    \label{eq:pickands-from-h}
    \tilde{A}(t)
    = \exp \left\{ \int_0^t \frac{\tilde{H}(z) - z}{z (1 - z)} \ dz \right\}
    \,.
\end{equation}
The estimator $\tilde{A}$ is not convex in general either.
\citeauthor{Jimenez2001} propose two modified versions of the CFG that satisfy the convexity constraint~\cite{Jimenez2001}.

Most estimation methods until the early 2010s are variants of either~\citeauthor{Pickands1981}', CFG or both~\cite{Vettori2017}.
More recent advances have focused on polynomials and splines.
For instance,~\citeauthor{Guillotte2016} study the conditions under which a polynomial, expressed in Bernstein form, is a \gls{pf}~\cite{Guillotte2016}.
\citeauthor{Marcon2017} use Bernstein-B\'ezier polynomials to enforce some \gls{pf} constraints~\cite{Marcon2017}.
\citeauthor{Cormier2014} use constrained quadratic smoothing B-splines to develop a compliant \gls{pf} in a nonparametric fashion using the \textsf{R} \texttt{cobs} package~\cite{Cormier2014}.

Previously,~\citeauthor{Einmahl2009} had introduced a compliant nonparametric estimator requiring constrained optimization and targeting an equivalent definition of \glspl{pf}~\cite{Einmahl2009}.
The \gls{pf} can be expressed~\cite{Guillotte2016} as
\begin{equation*}
    A(t) = \int_0^1 \max \{ t(1 - z), z(1 - t) \} \ d\mathcal{H}(z)\,,
\end{equation*}
where $\mathcal{H}$ is the so-called \textit{spectral} measure on $[0, 1]$: a finite measure satisfying $\int_0^1 z \ d\mathcal{H}(z) = 1$.
Under absolute continuity of $A'$~\cite{Guillotte2016}, $\mathcal{H}$ admits a decomposition
\begin{equation}
    \label{eq:spectral-measure-decomposition}
    \mathcal{H}(B) = \mathcal{H}_0 \mathbf{1}_B(0) + \int_B \eta(z) \ dz +
    \mathcal{H}_1 \mathbf{1}_B(1)\,,
\end{equation}
where $\mathbf{1}_B$ denotes the indicator function on $B$, $\eta = A''$ almost everywhere on $(0, 1)$, $\mathcal{H}_0 = 1 + A'(0^+)$ and $\mathcal{H}_1 = 1 - A'(1^-)$.

The concept of Williamson's transform has recently irrupted in copula theory~\cite{Bacigal2017}.
\citeauthor{McNeil2009} employ it in their study of $d$-monotone Archimedean generators~\cite{McNeil2009, McNeil2010}.
\citeauthor{Charpentier2014} also use it to model multivariate Archimax copulas~\cite{Charpentier2014}.
Even though they do not consider it in their work,~\citeauthor{Fontanari2020} introduce a subclass of Archimedean copulas called Lorenz copulas, where Williamson's transform could play a crucial role in estimation, as we later specify.

\paragraph{Goals}

Some accepted methods fail to meet all the constraints required by the \gls{pf}, even in the bivariate case~\cite{Vettori2017}.
Semiparametric approaches, like the one introduced by~\citeauthor{HernandezLobato2011} for Archimedean copulas~\cite{HernandezLobato2011}, have not been explored in the context of \glspl{evc}.

The research community is currently focusing on $n$-variate extensions~\cite{Gudendorf2012}.
However, a more flexible and sound construction is missing in the bivariate context.
The work by~\citeauthor{Kamnitui2019} suggests that the bivariate \gls{evc} family is not as narrow, especially under asymmetry~\cite{Kamnitui2019}.
Our method will thus exclusively focus on the bivariate setting.

The semiparametric procedure we introduce here offers the following advantages over state-of-the-art methods:
\begin{itemize}[label=$\diamond$, leftmargin=*]
    \setlength\itemsep{0ex}
    \item Built-in \gls{pf} constraint compliance, mapping any $\boldsymbol{\theta} \in \mathbb{R}^{d}$, for $d$ large, to some \gls{pf} $A_{\boldsymbol{\theta}}$ ranging in a broad spectrum of dependence strengths and asymmetries.
    \item Optimization of $\boldsymbol{\theta}$ via \gls{mle}, taking the most advantage of each observation, even in small samples.
    \item Ability to penalize model complexity during the optimization process, especially for large $d$, reducing overfitting and opening opportunities for Bayesian analysis.
\end{itemize}

The approach by~\citeauthor{Cormier2014}~\cite{Cormier2014}, also focusing on the bivariate case, allows for a large $d$, but does not retain control over the domain of $\boldsymbol{\theta}$.
Given a sample $\{(U_i, V_i)\}_{i = 1}^n$ from an \gls{evc} with \gls{pf} $A$, the \glspl{rv} $(Z_i, T_i)$, where $Z_i = \log U_i / \log (U_i V_i)$, $T_i = \log \hat{C}(U_i, V_i) / \log (U_i V_i)$, and $\hat{C}$ is the empirical estimate of $C$, lie close to $A$'s graph.
Then one can perform a constrained B-splines regression on those points.
However, the estimation procedure chooses the coordinates to satisfy the \gls{pf} constraints since not all parameters would be valid.
Hence, their method lacks a proper structure, falling into the nonparametric category.
Difficulties are bound to appear with small samples after relying on the empirical copula $\hat{C}$ and a regression approach.

\paragraph{Outline}

We introduce in Section~\ref{sec:method} our semiparametric method.
We formally construct and estimate a large subclass of~\glspl{evc} and explore their properties.
\ref{sec:proofs} includes all the proofs.
We then test our method on a \gls{ss} and a real-world case study in Section~\ref{sec:results}.
Section~\ref{sec:discussion} provides further comments on our method's performance and general possibilities.
Finally, we offer some concluding remarks in Section~\ref{sec:conclusions}.

    \section{Method}

\label{sec:method}

In the following sections, we will cover (i) the construction of a new semiparametric \gls{evc}, (ii) some of its properties, (iii) estimation algorithms, (iv) simulation, and (v) a possible solution to one of its limitations.

\subsection{Construction}

A copula arising from our construction will be called a \gls{sbevc}.
We will also refer to our method as \gls{sbevc}.
The construction of \glspl{sbevc} encompasses several steps.
The following sections will go through them from our \gls{pf} goal to a coordinates vector.
In each stage, the complexity of the parameter decreases, from an infinite-dimensional functional parameter with stringent constraints to an arbitrary $\boldsymbol{\theta} \in \mathbb{R}^d$.
That is not the natural order in the estimation phase, but it constitutes the safest path to weigh the sacrifices we make along the way.
Notwithstanding, we briefly summarize the journey in its final form:
\begin{enumerate}[leftmargin=*]
    \setlength\itemsep{0ex}
    \item Given $\boldsymbol{\theta} \in \mathbb{R}^d$, we build a \gls{zbs} $p_{\bm{\theta}}$ defined in $[0, 1]$.
    \item We apply to $p_{\bm{\theta}}$ the inverse \gls{clr} transformation to obtain a \gls{pdf} $f_{\bm{\theta}}$ supported on $[0, 1]$.
    \item We integrate $f_{\bm{\theta}}$ using the \gls{wt} to obtain a 2-monotone function $W_{\bm{\theta}}$ supported on $[0, 1]$.
    \item We affinely transform $W_{\bm{\theta}}$ to arrive at the \gls{pf} $A_{\bm{\theta}}$, as depicted in \figurename~\ref{fig:pickands-and-williamson}.
\end{enumerate}

\subsubsection{Affine transformation}

The support lines of the \gls{pf} resemble a pair of coordinate axes if rotated and scaled.
Let $M$ be the unique 2-dimensional affine transformation mapping $(0, 1)$, $(0, 0)$ and $(1, 0)$ to $(0, 1)$, $(1/2, 1/2)$ and $(1, 1)$, respectively.
$M$ and $M^{-1}$ take the form
\begin{multicols}{2}
    \noindent
    \begin{equation}
        \label{eq:affine-map}
        M(x, w) =
        \frac{1}{2}
        \begin{pmatrix}
            1 + x - w \\
            1 + x + w
        \end{pmatrix}
        \,,
    \end{equation}
    \columnbreak
    \begin{equation}
        \label{eq:affine-map-inverse}
        M^{-1}(t, a) =
        \begin{pmatrix}
            t + a - 1 \\
            a - t
        \end{pmatrix}
        \,.
    \end{equation}
\end{multicols}
\vspace*{-\multicolsep}
Under certain conditions, the inverse mapping $M^{-1}$ transforms the graph of a \gls{pf}, $\left\{ (t, A(t)) \ | \ t \in [0, 1] \right\}$, into the graph of a 2-monotone function $W$ defined on $[0, 1]$ and satisfying $W(0) = 1$ and $W(1) = 0$.
Here, 2-monotone stands for non-increasing and convex~\cite{McNeil2009, Charpentier2014}.
\figurename~\ref{fig:pickands-and-williamson} shows the transition from one to the other.
Note that for a twice differentiable function $W$ with the above boundary constraints, 2-monotonicity is equivalent to $W'(x) \leq 0$ and $W''(x) \geq 0$, for all $x \in (0, 1)$.

\begin{proposition}
    \thlabel{prop:A-W-equivalence}
    Let all the upcoming functions be twice differentiable on $(0, 1)$.
    By means of~\eqref{eq:affine-map} and its inverse~\eqref{eq:affine-map-inverse}, there is a one-to-one correspondence between \glspl{pf} $A$ satisfying $A(t) > 1 - t$, for all $t \in (0, 1/2]$, and 2-monotone functions $W$ defined on $[0, 1]$ and satisfying $W(0) = 1$, $W(1) = 0$.
    Namely, $A$ can be obtained from $W$ as
    \begin{equation}
        \label{eq:rotation}
        \left\{
            \begin{aligned}
                t(x) &= \frac{1}{2} \left( 1 + x - W(x) \right) \\
                A(t(x)) &= \frac{1}{2} \left( 1 + x + W(x) \right)
            \end{aligned}
        \right.
    \end{equation}
    and conversely, $W$ from $A$, as
    \begin{equation}
        \label{eq:rotation-inverse}
        \left\{
            \begin{aligned}
                x(t) &= t + A(t) - 1 \\
                W(x(t)) &= A(t) - t
            \end{aligned}
        \right.\,,
    \end{equation}
    where both $t(x)$ and $x(t)$ are automorphisms of $[0, 1]$.
    Moreover,
    \begin{multicols}{2}
        \noindent
        \begin{equation}
            \label{eq:pickands-first-derivative}
            A'(t(x)) = \frac{1 + W'(x)}{1 - W'(x)} \,,
        \end{equation}
        \columnbreak
        \begin{equation}
            \label{eq:pickands-second-derivative}
            A''(t(x)) = \frac{4 \ W''(x)}{\left( 1 - W'(x) \right)^3} \,,
        \end{equation}
    \end{multicols}
    \vspace*{-\multicolsep}
    \noindent and
    \begin{multicols}{2}
        \noindent
        \begin{equation}
            \label{eq:pickands-first-derivative-inverse}
            W'(x(t)) = \frac{A'(t) - 1}{A'(t) + 1} \,,
        \end{equation}
        \columnbreak
        \begin{equation}
            \label{eq:pickands-second-derivative-inverse}
            W''(x(t)) = \frac{2 \ A''(t)}{\left( 1 + A'(t) \right)^3} \,.
        \end{equation}
    \end{multicols}
    \vspace*{-\multicolsep}
\end{proposition}

\begin{remark}
    We do not need the smoothness assumption in~\thref{prop:A-W-equivalence} in either direction.
    We can argue that affine transformations map convex epigraphs into convex epigraphs\footnote{A function is convex \gls{iff} its epigraph is a convex set.}.
    Hence, the convexity of $A$ is equivalent to that of $W$.
    Nonetheless, differentiability is a convenient requirement for our construction.
\end{remark}

\begin{example}
    \thlabel{ex:polynomial-pickands-W}
    Elaborating on \figurename~\ref{fig:pickands-and-williamson}, plugging $A(t) = t^2 - t + 1$ into~\eqref{eq:rotation-inverse} yields the $W(x) = x - 2 \sqrt{x} + 1$ we see in \figurename~\ref{fig:williamson}.
\end{example}

\begin{example}
    \thlabel{ex:power-W-family}
    The family of functions $W_{\theta}(x) = (1 - x)^{\theta}$, where $\theta \in [1, \infty)$, meet the conditions in \thref{prop:A-W-equivalence} and thus produce \glspl{evc}.
\end{example}

A $W$ function like the one defined in~\thref{prop:A-W-equivalence} induces a spectral measure~\eqref{eq:spectral-measure-decomposition} by means of
\begin{equation}
    \label{eq:W-spectral-density}
    \eta(z) =
    \frac{4 \ W''(t^{-1}(z))}{\left[ 1 - W'(t^{-1}(z)) \right]^3}
\end{equation}
and
\begin{multicols}{2}
    \noindent
    \begin{equation}
        \label{eq:W-spectral-point-mass-0}
        \mathcal{H}_0 = \frac{2}{1 - W'(0^+)} \,,
    \end{equation}
    \columnbreak
    \begin{equation}
        \label{eq:W-spectral-point-mass-1}
        \mathcal{H}_1 = \frac{-2 W'(1^-)}{1 - W'(1^-)} \,.
    \end{equation}
\end{multicols}
\vspace*{-\multicolsep}
\noindent Such a $W$ fails to attain the comonotonic copula, which has \gls{pf} $A(t) = \max \{ 1 - t, t \}$.
However, it can still model independence if $W(x) = 1 - x$.

\subsubsection{Williamson's transform}

Transitioning from $A$ to $W$ is cheap.
However, $W$ still poses stringent constraints on derivatives and boundary conditions.
We can solve them by taking $W$ as the \gls{wt} of a \gls{rv} supported on $[0, 1]$ that places no mass at zero.

\begin{definition}[Williamson's transform]
    Let $F$ be the \gls{cdf} of a non-negative \gls{rv} satisfying $F(0) = 0$.
    We define the \textit{\gls{wt}} of $F$ as
    \begin{equation*}
        \mathfrak{W}\{F\}(x) =
        \int_x^{\infty} \left( 1 - \frac{x}{r} \right) dF(r) \,.
    \end{equation*}
\end{definition}

A fundamental result in~\cite{McNeil2009} states that $\Psi = \mathfrak{W}\{F\}$ \gls{iff} $\Psi$ is 2-monotone and satisfies the boundary conditions $\Psi(0) = 1$ and $\Psi(\infty) = \lim_{x \rightarrow \infty} \Psi(x) = 0$.
Moreover, such an $F$ is unique and can be retrieved from $\Psi$ as $F(x) = 1 - \Psi(x) + x \ \Psi'(x^{+})$.
It can be easily checked that the support of $F$ is $[0, x^*]$, where $x^* = \inf \{ x \in \mathbb{R} \cup \{ \infty \} \ \vert \ \Psi(x) = 0 \}$.
In our case, the support is bounded, since $W(1) = 0$.
Therefore, we get the following corollary.

\begin{corollary}
    \thlabel{cor:williamson-characterization}
    A function $W: [0, 1] \rightarrow \mathbb{R}$ is 2-monotone\footnote{Non-negative, non-increasing and convex.} with $W(0) = 1$ and $W(1) = 0$ \gls{iff} it can be expressed as
    \begin{equation*}
        W(x) = \int_x^1 \left( 1 - \frac{x}{r} \right) dF(r) \,,
    \end{equation*}
    for some unique \gls{cdf} $F$ supported on $[0, 1]$ and such that $F(0) = 0$.
\end{corollary}

We can further simplify the construction of $W$ by imposing $F$ to be absolutely continuous with \gls{pdf} $f$:
\begin{equation}
    \label{eq:williamson}
    W(x) = \int_x^1 \left( 1 - \frac{x}{r} \right) f(r) \ dr \,.
\end{equation}
The form~\eqref{eq:williamson} adds smoothness to $W$.
Differentiating~\eqref{eq:williamson} we get
\begin{multicols}{2}
    \noindent
    \begin{equation}
        \label{eq:williamson-first-derivative}
        W'(x) = -\int_x^1 \frac{f(r)}{r} \ dr \,,
    \end{equation}
    \columnbreak
    \begin{equation}
        \label{eq:williamson-second-derivative}
        W''(x) = \frac{f(x)}{x} \,.
    \end{equation}
\end{multicols}
\vspace*{-\multicolsep}
\noindent All in all, the $W$ function satisfies the equation
\begin{equation}
    \label{eq:W-survival}
    W(x) = \hat{F}(x) + x \ W'(x) \,,
\end{equation}
where $\hat{F}(x) = 1 - F(x) = \int_x^1 f$ is the survival function of $F$.
Equation~\eqref{eq:W-survival} is useful for computational purposes.
From~\eqref{eq:williamson-first-derivative}, it directly follows $W'(1^{-}) = 0$, thus $\mathcal{H}_1$ in~\eqref{eq:W-spectral-point-mass-1} is equal to zero.
This feature prevents \gls{sbevc} from reaching the independence copula, for which $W(x) = 1 - x$.
The value of $W'(0^{+})$ (and subsequently of $\mathcal{H}_0$) is, however, dependant on the behaviour of $f$ near zero.

\begin{example}
    \thlabel{ex:polynomial-pickands-F}
    Expanding on \thref{ex:polynomial-pickands-W}, by using~\eqref{eq:williamson-second-derivative}, we find that $F(x) = \sqrt{x}$.
    Hence, $F$ is the \gls{cdf} of $U^2$, where $U \sim \text{Unif}[0, 1]$.
\end{example}

\begin{example}
    \thlabel{ex:power-W-beta}
    Elaborating on \thref{ex:power-W-family}, if $\theta > 1$, by~\eqref{eq:williamson-second-derivative}, we get $f(x) = \theta (\theta - 1) x (1 - x)^{\theta - 2}$, which is the \gls{pdf} of the $\text{Beta}(\alpha = 2, \beta = \theta - 1)$ distribution.
\end{example}

\begin{example}
    \thlabel{ex:W-power-densities}
    \thref{ex:polynomial-pickands-F} is a special case of \glspl{wt} of positive\footnote{For $\theta \leq 0$, the resulting \gls{rv} is not bounded.} powers $U^{\theta}$ of the uniform distribution on $[0, 1]$.
    The general formulas for their densities and \glspl{cdf} are $f_{U^{\theta}}(x) = x^{1 / \theta - 1} / \theta$ and $F_{U^{\theta}}(x) = x^{1 / \theta}$, respectively, whereas their \glspl{wt} are given by
    \begin{equation}
        \label{eq:williamson-uniform-densities}
        W_{U^{\theta}}(x) =
        \left\{
            \begin{aligned}
                1 + \frac{1}{\theta - 1} x
                - \frac{\theta}{\theta - 1} x^{\frac{1}{\theta}} \,,
                    & \ \text{if} \ \theta \neq 1 \\
                1 - x + x \log x \,,
                    & \ \text{if} \ \theta = 1
            \end{aligned}
        \right.
        \,.
    \end{equation}
\end{example}

\subsubsection{Bayes space}

For modelling $f$, we will resort to the Bayes space, i.e., the Hilbert space $(\mathcal{B}^2, \oplus, \odot)$ of probability density functions of square-integrable logarithm~\cite{Egozcue2006, Machalova2020}.
The space $\mathcal{B}^2$ can be injected into $L^2([0, 1])$ employing the \gls{clr} transformation
\begin{equation}
    \label{eq:clr}
    \text{clr}[f](x) = \log f(x) - \int_0^1 \log f(y) \ dy \,.
\end{equation}
However, not every element in $L^2([0, 1])$ is attainable, since~\eqref{eq:clr} introduces the constraint $\int_0^1 \text{clr}[f] = 0$.
If we define the subspace $L_0^2([0, 1])$ of the functions with zero integral, then~\eqref{eq:clr} is a bijection from $\mathcal{B}^2$ to $L_0^2([0, 1])$ with inverse
\begin{equation}
    \label{eq:clr-inverse}
    \text{clr}^{-1}[p](x) = \frac{\exp p(x)}{\int_0^1 \exp p(y) \ dy} \,.
\end{equation}
What is more,~\eqref{eq:clr} is an isometry between $\mathcal{B}^2$ and $L_0^2([0, 1])$.

\begin{example}
    \thlabel{ex:clr-uniform-densities}
    The densities of positive powers $U^{\theta}$ of the uniform distribution in \thref{ex:W-power-densities} have \gls{clr} transforms $\text{clr}[U^{\theta}](x) = (1 - \theta)(1 + \log x) / \theta$.
    It immediately follows that all $U^{\theta}$ are linearly dependent.
\end{example}

Utilizing the isometry~\eqref{eq:clr-inverse}, we can search for a suitable function in $L_0^2([0, 1])$ and then transform it back to a \gls{pdf}.
However, this space is infinite-dimensional.
In practice, we shall work on a finite subspace.
In general, we will build a \gls{pdf} $f_{\bm{\theta}}$ as a linear combination
\begin{equation}
    \label{eq:bayes-model}
    f_{\bm{\theta}}(x) =
    \bigoplus_{i = 1}^{n} \ (\theta_i \odot \text{clr}^{-1}[\varphi_i])(x)
    =
    \text{clr}^{-1}
    \left[ \sum_{i = 1}^{n} \ \theta_i \ \varphi_i(x) \right]
    \,,
\end{equation}
where we can assume the $(\varphi_i)_{i = 1}^n$ are orthonormal, i.e., $\langle \varphi_i, \varphi_j \rangle_{L^2([0, 1])} = \delta_{ij}$, the Kronecker delta, and satisfy the zero-integral constraint.

The null element $f_{\bm{0}} \sim \text{Unif}[0, 1]$ produces the \gls{wt}~\eqref{eq:williamson-uniform-densities} for $\theta = 1$.
After rotation~\eqref{eq:rotation}, the resulting \gls{pf} has an explicit form that involves the Lambert W function~\cite{Corless1996}.
Should the parameters in~\eqref{eq:bayes-model} be normally distributed with zero mean vector, we expect the \gls{pf} to lie close to the graph in \figurename~\ref{fig:lambert-pickands}.
In this sense, \gls{sbevc} presents a very slight bias towards asymmetry.
This bias can be corrected by considering an affine subspace instead of a pure vector space, using a convenient $\omega \in L^2([0, 1])$ as a centre:
\begin{equation}
    \label{eq:affine-bayes-model}
    f_{\bm{\theta}} =
    \text{clr}^{-1}[\omega] \oplus
    \bigoplus_{i = 1}^{n} \ (\theta_i \odot \text{clr}^{-1}[\varphi_i])
    \,.
\end{equation}

\begin{figure}[ht]
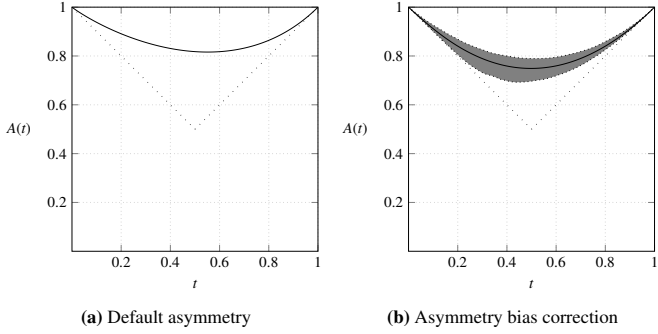

    \centering
    \begin{subfigure}{\figurewidth}
        \centering
        \resizebox{\columnwidth}{!}{
            \input{tikz/02-a-lambert.tex}
        }
        \caption{Default asymmetry}
        \label{fig:lambert-pickands}
    \end{subfigure}%
    \begin{subfigure}{\figurewidth}
        \centering
        \resizebox{\columnwidth}{!}{
            \input{tikz/02-b-bias-interval.tex}
        }
        \caption{Asymmetry bias correction}
        \label{fig:bias-interval}
    \end{subfigure}%
    \caption{
        On the left, \gls{pf} $A(t) = 1 - t + \exp \{ \mathcal{W}_{-1}(-2t / e^2) + 2 \}$ arising from the Williamson family~\eqref{eq:williamson-uniform-densities} for $\theta = 1$.
        Here, $\mathcal{W}_{-1}$ denotes the $k = -1$ branch of the complex Lambert W function.
        On the right, random \glspl{pf} built as perturbations around centre~\eqref{eq:orthogonal-projection}.
        All $\theta_i$ were sampled from a normal distribution with zero mean and $\sigma = 0.1$.
        We randomly drew a total of 1,000 \glspl{pf}.
        The solid line represents the mean function, while the grey envelope represents the confidence interval between quantiles 1\% and 99\%.
        The mean line is close to the $A(t) = t^2 - t + 1$ in \figurename~\ref{fig:pickands}.
    }
\end{figure}

\subsubsection{Compositional splines}

\citeauthor{Machalova2020} in~\cite{Machalova2020} formalize the construction of a compliant \gls{zbs} as a linear combination of the usual B-splines.
We shall approximate the \gls{clr} space with the \gls{zbs} subspace.
We refer the reader to~\cite{Machalova2020} for further details on how to compute \glspl{zbs} and~\cite{Boor2002} for more profound knowledge on B-splines, in general.

Splines are bounded functions.
Committing to them, we would definitely have $W'(0^{+}) = -\infty$ and thus $\mathcal{H}_0 = 0$ in~\eqref{eq:W-spectral-point-mass-0}.
Therefore, the resulting spectral measure would be absolutely continuous with respect to the Lebesgue measure on $[0, 1]$ with Radon-Nikodym derivative equal to~\eqref{eq:W-spectral-density}.

Given $0 = \kappa_0 < \dots < \kappa_{n + 1} = 1$, where $n \geq 0$, and assuming $2d$ additional coincidental\footnote{Coincidental knots at the interval endpoints convey maximum smoothness at each interior knot~\cite{Boor2002}. For splines of degree less than or equal to $d$, we have $(d - 1)$-continuous differentiability everywhere in $[0, 1]$.} knots $\kappa_{-d} = \dots = \kappa_{-1} = 0$ and $\kappa_{n + 2} = \dots = \kappa_{n + d + 1} = 1$ at the endpoints, the space of splines $p \in \mathcal{Z}_{\bm{\kappa}}^d$ of degree less than or equal to $d$ and $n + 2$ different knots $\bm{\kappa} = (\kappa_i)_{i = 0}^{n + 1}$ has dimension $n + d$.
The case $n = 0$ corresponds to zero-integral polynomials over $[0, 1]$.
Altogether, any \gls{zbs} can be expressed as
\begin{equation}
    \label{eq:zb-spline-model}
    p_{\bm{\theta}}(x) = \sum_{i = 1}^{n + d} \theta_i Z_i(x) \,,
    \ \text{for} \
    \bm{\theta} = (\theta_1, \dots, \theta_{n + d})
    \in \mathbb{R}^{n + d} \,,
\end{equation}
where $\int_0^1 Z_i = 0$ and we can further assume an orthonormal basis~\cite{Machalova2020}, i.e., $\langle Z_i, Z_j \rangle_{L^2} = \delta_{ij}$.

Furthermore, we can place a convenient centre for our \gls{zbs} space to correct the asymmetry bias.
We propose to take $\omega$ in~\eqref{eq:affine-bayes-model} to be the orthogonal projection $z$ of $-(1 + \log x) / 2$, the case $\theta = 2$ in~\thref{ex:clr-uniform-densities}, onto the space~\eqref{eq:zb-spline-model}:
\begin{equation}
    \label{eq:orthogonal-projection}
    z(x) =
    \sum_{i = 1}^{n + d}
    \left\langle \text{clr}[U^2], Z_i \right\rangle_{L^2} Z_i(x)
    \,.
\end{equation}
\figurename~\ref{fig:orthogonal-projection} shows that a spline can effectively approximate the logarithmic centre, despite the divergence near zero.
\figurename~\ref{fig:zb-spline-basis} depicts the underlying orthonormal \gls{zbs} basis $\{Z_i\}$.
\figurename~\ref{fig:bias-interval} shows the effectiveness of the bias correction.

\begin{figure}[ht]
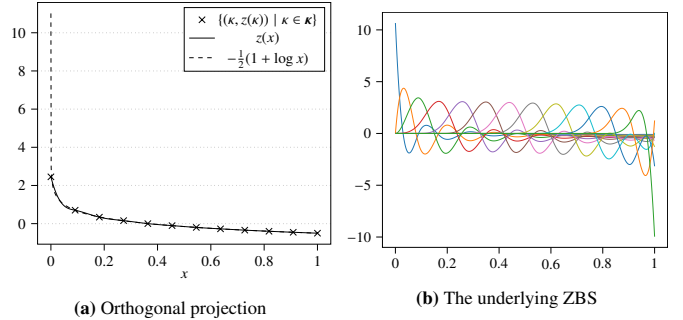

    \centering
    \begin{subfigure}{\figurewidth}
        \centering
        \resizebox{\columnwidth}{!}{
            \input{tikz/03-a-orthogonal-projection.tex}
        }
        \caption{Orthogonal projection}
        \label{fig:orthogonal-projection}
    \end{subfigure}%
    \begin{subfigure}{\figurewidth}
        \centering
        \resizebox{\columnwidth}{!}{
            \input{tikz/03-b-zbspline-basis.tex}
        }
        \caption{The underlying \gls{zbs}}
        \label{fig:zb-spline-basis}
    \end{subfigure}
    \caption{
        On the left, the projection $z(x)$ of $-(1 + \log x) / 2$ onto an orthonormal \gls{zbs} 13-dimensional basis with knots $\bm{\kappa}$.
        The logarithmic function diverges to infinity at zero.
        On the right, we have the underlying orthonormal cubic \gls{zbs} basis with 13 elements.
    }
\end{figure}

\subsection{Properties}

The following sections provide some insights on the relation between the core \gls{pdf} and the resulting \gls{evc}.

\subsubsection{Convergence}

We will present some results on how convergence on the Bayes space relates to convergence for the resulting \glspl{evc} through \gls{sbevc}.
We shall use the supremum norm $\lVert f \rVert_{\infty} = \sup_{x \in \mathcal{X}} \ \lvert f(x) \rvert$ of a bounded function $f: \mathcal{X} \rightarrow \mathbb{R}$ to measure distances between some objects.
$\mathcal{X}$ will typically be a compact subset of $\mathbb{R}^n$, namely $[0, 1]$, for functions $W$ and $A$, and $[0, 1]^2$, for copulas $C$.
The supremum norm defines a distance $d_{\infty}(f, g) = \lVert f - g \rVert_{\infty}$.
A sequence of functions $\{f_n\}_{n = 1}^\infty$ converging on the latter distance to some $f$ is said to converge \textit{uniformly}.
Sometimes, however, uniform convergence is a too strong property.
For instance, the function $f$ may not be bounded on the whole $\mathcal{X}$.
Another convergence exists in those cases, only requiring the sequence converging uniformly to $f$ on every compact subset $\mathcal{K} \subset \mathcal{X}$.
Then, the sequence of $\{f_n\}_{n = 1}^\infty$ is said to be \textit{compactly convergent} to $f$.
The following related concept applies to probability measures.

\begin{definition}
    \thlabel{def:total-variation}
    Let $\mathbb{P}$ and $\mathbb{Q}$ be probability measures on the space $[0, 1]$ equipped with the Borel $\sigma$-algebra $\mathcal{B}$.
    The \textit{\gls{tvd}} between $\mathbb{P}$ and $\mathbb{Q}$ is defined as
    \begin{equation*}
        d_{\text{TV}}(\mathbb{P}, \mathbb{Q}) =
        \sup_{B \in \mathcal{B}} \ \lvert \mathbb{P}(B) - \mathbb{Q}(B) \rvert
        \,.
    \end{equation*}
\end{definition}

TVD satisfies all three axioms of a proper metric.
By Scheff\'e's theorem~\cite{Tsybakov2009}, it can also be expressed in terms of the \glspl{pdf} $f$ and $g$ of $\mathbb{P}$ and $\mathbb{Q}$, respectively, as
\begin{equation}
    \label{eq:scheffe-theorem}
    d_{\text{TV}}(\mathbb{P}, \mathbb{Q}) \equiv
    d_{\text{TV}}(f, g) =
    \frac{1}{2} \int_0^1 \lvert f(x) - g(x) \rvert \ dx
    \,.
\end{equation}

Our first result links convergence in \gls{tvd} of a sequence of \glspl{pdf} with convergence of the corresponding sequence of \glspl{wt} and their derivatives.

\begin{proposition}
    \thlabel{prop:fn-converge-imply-wn-converge}
    Let $\{f_n\}_{n = 1}^\infty$ be a sequence of \glspl{pdf} supported on $[0, 1]$ such that $\lim\limits_{n \rightarrow \infty} d_{TV}(f, f_n) = 0$ for some \gls{pdf} $f$ also on $[0, 1]$.
    Let $W_n$ and $W$ be the corresponding Williamson transforms of $f_n$ and $f$, respectively.
    Then, the sequence $\{W_n\}_{n = 1}^\infty$ uniformly converges to $W$ on $[0, 1]$.
    Moreover, $\{W_n'\}_{n = 1}^\infty$ compactly converges to $W'$ on $(0, 1]$.
\end{proposition}

The next step links the uniform convergence of \glspl{wt} with that of \glspl{pf}.
Uniform convergence of pointwise-convergent sequences of \glspl{pf} can be established by other means~\cite{FilsVilletard2008}.
Nonetheless, the following result also states the uniform convergence of the first derivatives of \glspl{pf} under the same hypotheses, which cannot be taken for granted.

\begin{proposition}
    \thlabel{prop:graph-convergence}
    Let $\{W_n\}_{n = 1}^\infty$ be a sequence of \glspl{wt}, arising from \glspl{pdf}, uniformly convergent to some other \gls{wt} $W$ on $[0, 1]$.
    Also, suppose the sequence of first derivatives $\{W_n'\}_{n = 1}^\infty$ of the previous functions compactly converge to $W'$ on $(0, 1]$.
    Let $A_n$ and $A$ be the corresponding \glspl{pf} of $W_n$ and $W$, respectively, according to \gls{sbevc}.
    Then, the sequence $\{A_n\}_{n = 1}^{\infty}$ uniformly converges to $A$ on $[0, 1]$, while $\{A_n'\}_{n = 1}^{\infty}$ compactly converges to $A'$ on $(0, 1]$.
\end{proposition}

\begin{remark}
    Uniform convergence of function derivatives is mainly unconnected to uniform convergence of the functions themselves.
    The reason why it works, in this case, comes down to Williamson's transform, whose first derivative is also in a convenient integral form that allows applying \gls{tvd} convergence of the internal \glspl{pdf} to both the function and its derivative simultaneously.
\end{remark}

Some topological arguments allow establishing uniform convergence of copulas from pointwise convergence alone~\cite{Schmitz2003}.
Notwithstanding, there exists a connection between the supremum norms of copulas and those of their respective \glspl{pf}~\cite{Kamnitui2019}.
Namely, $\lVert C_1 - C_2 \rVert_{\infty} \leq 2\gamma / (1 + 2\gamma)^{1 + 1/(2\gamma)}$, where $\gamma = \lVert A_1 - A_2 \rVert_{\infty}$.

We turn next to some assumptions making convergence in $L_0^2([0, 1])$ sufficient for \glspl{pdf} to converge in \gls{tvd}.

\begin{proposition}
    \thlabel{prop:convergence-tvd-pdfs}
    Let $\{p_n\}_{n = 1}^\infty \subset L_0^2([0, 1])$ continuous and uniformly bounded, i.e., $\lVert p_n \rVert_{\infty} \leq K$ for some $K > 0$ and for all $n$.
    Suppose $\lim\limits_{n \rightarrow \infty} \lVert p - p_n \rVert_2 = 0$, for some $p \in L_0^2([0, 1])$.
    Let $f_n = \mathrm{clr}^{-1}[p_n]$ and $f = \mathrm{clr}^{-1}[p]$.
    Then, $\lim\limits_{n \rightarrow \infty} d_{TV}(f, f_n) = 0$.
\end{proposition}

\gls{sbevc}, acting on convergent sequences in $L_0^2([0, 1])$, produces \glspl{evc} that not only uniformly converge but also whose partial derivatives do.
Since copula partial derivatives correspond to conditional \glspl{cdf}, they are paramount in copula sampling algorithms~\cite{Bouye2000}.
\thref{cor:convergence-partial-derivatives} guarantees that the resulting samples tend to fit the copula uniformly over the support.

\begin{corollary}
    \thlabel{cor:convergence-partial-derivatives}
    Let $\{z_n\}_{n = 1}^\infty \subset L_0^2([0, 1])$ be a sequence of uniformly bounded smooth cubic \glspl{zbs} that converge in the $\lVert \cdot \rVert_2$ norm.
    The corresponding $\mathrm{\glspl{evc}}$ from \gls{sbevc} $\{C_n\}_{n = 1}^{\infty}$ uniformly converge to some $\mathrm{\gls{evc}}$ $C$, satisfying $\partial_i C_n \xrightarrow[n \rightarrow \infty]{\lVert \cdot \rVert_{\infty}} \partial_i C$ compactly over $(0, 1]^2$, for $i = 1, 2$.
\end{corollary}

\subsubsection{Dependence}

In the case of an \gls{evc}, Kendall's tau and Spearman's rho take the form of integrals involving the \gls{pf}~\cite{Kamnitui2019}.
The substitution of the \gls{pf} by the equivalent \gls{wt} form and a change of variables afterwards do not provide any meaningful insight into the latter's role.
However, the apparent relation between \glspl{wt} and Lorenz curves reveals a new path for measuring association.

The \gls{wt} $W$ of a \gls{pf} $A$ satisfies the definition of a Lorenz~\cite{Fontanari2020} curve $L$ after the change of variable $L(x) = W(1 - x)$, for $x \in [0, 1]$.
Lorenz curves appear in econometrics for assessing wealth inequality through the \gls{gc} $G = 1 - 2 \int_0^1 L$.
The $G$ index has a geometrical interpretation as the area between $L$ and $x \mapsto x$ divided by the area under $x \mapsto x$, which is equal to $1/2$.
The same interpretation applies to $W$ and, since $\int_0^1 W = \int_0^1 L$, we have $G = 1 - 2 \int_0^1 W$.
A value of $G = 0$ means wealth is uniformly distributed (the $p\%$ wealthiest proportion of the population accumulate $p\%$ of the total wealth, for all $p \in [0, 1]$), whereas $G \lesssim 1$ means that nearly all wealth belongs to a tiny fraction of the population.
In summary, $G = 0$ represents perfect equality, while $G = 1$ represents perfect inequality.

In our context, a \gls{pf} $A$ is the affine transformation of some $W$.
Since affine transformations change areas by applying a constant factor, the latter cancels out in ratio measures, leaving them invariant.
Therefore, the \gls{gc} is the area between $A$ and the upper bound line $\{y = 1\}$ divided by the area between the support lines and the upper bound line.
This argument leads to $G = 4 \left( 1 - \int_0^1 A \right)$.
The value $G = 1$ happens when $A$ is identically equal to the lower support lines ($\int_0^1 A = 3/4$), producing the comonotonic copula.
On the other hand, $G = 0$ occurs when $A$ is identical to the upper bound line ($\int_0^1 A = 1$), producing the independence copula.
This way, the bivariate positive association can be interpreted in econometric terms: comonotonicity is equivalent to perfect inequality, whereas independence corresponds to perfect equality.

The \gls{gc} is an uncommon association measure in the context of \glspl{evc}, despite its simplicity.
We have only found a brief mention of it in~\cite{Guzmics2020}, where the measure was not scaled to lie on $[0, 1]$.
Moreover, \gls{pf} symmetry was further assumed.
The index can be reformulated for any \textit{positively quadrant dependent} copula $C$, i.e., $C(u, v) \geq uv$, for all $u, v \in [0, 1]^2$, as\footnote{The \gls{cdf}~\eqref{eq:h-cdf} and its stochastic interpretation provide a shortcut to check this.}
\begin{equation*}
    G = 4 \left(
        1 - \int_{[0, 1]^2} \frac{\log C(u, v)}{\log uv} \ dudv
    \right)
    \,.
\end{equation*}

The \gls{gc}, as defined above, satisfies the axioms of a \textit{dependence measure}~\cite{Bouye2000}.
The \gls{gc} takes values on $[0, 1]$, unlike Kendall's tau and Spearman's rho, which belong to a more general family of \textit{concordance measures}, ranging in $[-1, 1]$ and allowing for negative association.

\begin{remark}
    Interestingly, after integrating by parts twice, the area of $W$ can be expressed in terms of the inner \gls{pdf} $f$, yielding $G = 1 - \mathbb{E}[X], \ \text{where} \ X \sim f$.
    This means that every $\mathrm{\gls{gc}}$ in $(0, 1)$ is attainable through $\mathrm{\gls{sbevc}}$.
\end{remark}
\subsection{Estimation}

\label{subsec:estimation}

The estimation process builds upon the various constructions explored above.
Given an orthonormal \gls{zbs} basis, we aim to find the parameter vector $\bm{\theta}$ that best fits a dataset.
Knowing the one-to-one relation~\cite{Caperaa1997} between the random vector $(U, V)$ following an \gls{evc} and the \gls{rv} $Z = \log U / \log (UV)$, we reduce our problem to fitting the latter, which has a more straightforward \gls{pdf}, derived from~\eqref{eq:h-cdf}:
\begin{equation}
    \label{eq:h-pdf}
    h(z) = 1 + (1 -2z) \frac{A'(z)}{A(z)} + z(1 - z)
    \left[ \frac{A''(z)}{A(z)} - \left( \frac{A'(z)}{A(z)} \right)^2 \right]
    \,.
\end{equation}

Given a random sample $\mathcal{D} = \{ z_i \}_{i = 1}^m$ from $Z$ and a model $h_{\bm{\theta}}$ derived from $p_{\bm{\theta}}$ up to $A_{\bm{\theta}}$, the frequentist approach to the estimation addresses the maximization of the \gls{pll} of $h_{\bm{\theta}}$
\begin{equation}
    \label{eq:loss-function}
    \ell(\bm{\theta} \vert \mathcal{D}) =
    \sum_{i = 1}^{m} \log h_{\bm{\theta}}(z_i)
    - \lambda \int_0^1 (p_{\bm{\theta}}''(x))^2 \ dx
    \, \ ,
\end{equation}
for some regularization hyper-parameter $\lambda \geq 0$.
The square norm term involving $p_{\bm{\theta}}''$ is the linearized curvature of the spline: a simplified non-intrinsic form of the curvature that can be expressed as a covariant tensor ${\bm{\theta}}^{\top} \Omega \bm{\theta}$, where $\Omega = (\Omega_{ij}) = (\int_0^1 Z_i'' Z_j'')$.
Splines may exhibit complex shapes prone to overfitting, as shown in \figurename~\eqref{fig:zb-spline-basis}.
Penalizing the curvature is the proposed method in~\cite{Machalova2020} in the context of compositional data regression.
\citeauthor{HernandezLobato2011} applied this approach in semiparametric copula models before~\cite{HernandezLobato2011}.
Taking $\lambda = 0$ removes regularization, retrieving the usual log-likelihood.

Estimating the parameters of such a model poses some challenges.
Evaluating the resulting \gls{pf} from a parameter vector and a single argument implies several non-trivial operations, most notably the integral and affine transformations~\eqref{eq:williamson} and~\eqref{eq:rotation}.
These steps can be applied with near-perfect accuracy, with proper algorithms and without time constraints.
However, in an iterative optimization, time is scarce.
Therefore, we propose critical approximations at each step that trade some accuracy off for processing speed without compromising the overall stability.
The effectiveness of our proposal will be thoroughly tested in a \gls{ss} in Section~\ref{sec:results}.

First, note that the evaluation of $A$ in~\eqref{eq:rotation} at a specific value $t_0$ requires solving for $x$ in $t(x) = t_0$.
The latter will generally be a nonlinear equation that can only be solved through numerical methods at a relatively high computational cost.
Hence, in most cases, evaluating~\eqref{eq:h-pdf} in~\eqref{eq:loss-function} at each point $z_i$ becomes rapidly unaffordable as $m$ increases.
Moreover, any root finding procedure would prevent us from applying gradient optimization, stopping backpropagation.
We propose $h_{\bm{\theta}}$ be approximated by a piecewise linear interpolator $\tilde{h}$ with sufficiently numerous and carefully selected knots.

Since~\eqref{eq:loss-function} is based on an empirical univariate sample $\mathcal{D}$, a good knot selection utilizes uniform quantiles of $\mathcal{D}$.
This way, the knots will be more spaced on low probability regions and accumulate on high probability ones.
This criterion, which was employed in a similar setting in~\cite{HernandezLobato2011}, reduces the variance of the parameter vector $\bm{\theta}$.
Once fixed the quantiles $\{q_i\}_{i = 1}^k$, we need to estimate some $\{x_i\}_{i = 1}^k$ such that $t_i \equiv t_{\bm{\theta}}(x_i) \approx q_i$ and then take $h_i = h_{\bm{\theta}}(t_i)$ as the linear interpolator value at knot $t_i$.
Note that the $t_i$'s are approximations for the $q_i$'s.
To estimate the required $x_i$'s, we may apply~\eqref{eq:rotation-inverse} over the $q_i$'s grid using an empirical nonparametric estimate of the \gls{pf}, like~\eqref{eq:pickands-from-h}.
We can state the procedure as follows.

\begin{algorithm}[Selection of an interpolation grid for $h_{\bm{\theta}}$]
    \thlabel{alg:empirical-w-grid}
    Let $\mathcal{D} = \{z_i\}_{i = 1}^m$ be a random sample following the $H$ distribution.
    To build an interpolation grid $\{x_i\}_{i = 0}^{k + 1}$ in the $W$ space such that $\{t(x_i)\}_{i = 0}^{k + 1}$ are roughly distributed according to $\mathcal{D}$, follow these steps:
    \begin{enumerate}[noitemsep]
        \item Pick $k$ uniform quantiles $0 = q_0 < \dots < q_{k + 1} = 1$ of $\mathcal{D}$.
        \item Build the empirical \gls{cdf} $\tilde{H}$ of $\mathcal{D}$.
        \item Build an empirical estimate $\tilde{A}$ using $\tilde{H}$ and~\eqref{eq:pickands-from-h}.
        \item Ensure boundary constraints taking
        \begin{equation*}
            \hat{A}(t) = \min \{1, \max \{ t, 1 - t, \tilde{A}(t)\} \}
            \,.
        \end{equation*}
        \item Set $x_0 = 0$ and $x_{k + 1} = 1$.
        Then, for every $i \in \{ 1, \dots, k \}$, set $x_i = q_i + \hat{A}(q_i) - 1$.
        \item Sort ascendingly the resulting $\{x_i\}_{i = 0}^{k + 1}$ and remove duplicates if needed.
    \end{enumerate}
\end{algorithm}

We can reuse the grid obtained in the last algorithm throughout the estimation process, at every gradient descent step and with different values for the parameter vector $\bm{\theta}$.
With this grid and a parameter vector $\bm{\theta}$, we can now build a light version of $h_{\bm{\theta}}$ to evaluate the \gls{pll}.
Early experiments suggest that selecting spline knots for $p_{\bm{\theta}}$ according to~\thref{alg:empirical-w-grid} is key to constructing an unbiased estimator $A_{\bm{\theta}}$.

Along with the interpolation grid, we need to estimate the values of $W_{\bm{\theta}}$ and its derivatives from $p_{\bm{\theta}}$.

\begin{algorithm}[Approximation of the \gls{wt} and its derivatives]
    \thlabel{alg:estimate-W}
    Let $p_{\bm{\theta}}$ be the \gls{zbs} corresponding to the parameter vector $\bm{\theta}$.
    Let $\{r_i\}_{i = 0}^{n + 1}$ be an strictly increasing real sequence such that $r_0 = 0$ and $r_{n + 1} = 1$.
    Let $\epsilon \gtrsim 0$ such that $\epsilon < r_1$.
    To build an approximation to the corresponding \gls{wt} $W_{\bm{\theta}}$ and its first and second derivatives, follow these steps:

    \begin{enumerate}[noitemsep]
        \item For $i \in \{ 0, \dots, n + 1 \}$, set $p_i = \exp p_{\bm{\theta}}(r_i)$.
        \item Compute $I \approx \int_0^1 \exp p_{\bm{\theta}}$ using the composite trapezoidal rule over $\{(r_i, p_i)\}_{i = 0}^{n + 1}$.
        \item For $i \in \{ 0, \dots, n + 1 \}$, set $f_i = p_i / I$.
        \item Set $s_0 = \epsilon$.
        Then, for $i \in \{ 1, \dots, n + 1 \}$, set $s_i = r_i$.
        \item For $i \in \{ 0, \dots, n + 1 \}$, set $\bar{W}_i'' = f_i / s_i$.
        \item For $i \in \{ 0, \dots, n \}$, set $\Delta_i = s_{i + 1} - s_i$.
        \item For $i \in \{ 0, \dots, n \}$, set $P_i = \Delta_i \ (f_i + f_{i + 1}) / 2$ and $Q_i = \Delta_i \ (\bar{W}_i'' + \bar{W}_{i + 1}'') / 2$.
        \item Set $\bar{W}_{n + 1}' = 0$.
        Then, for $i$ from $n$ down to $0$, compute $\bar{W}_i'$ using the recurrence relation
        \begin{equation}
            \label{eq:recurrence-W-derivative}
            \bar{W}_i' = \bar{W}_{i + 1}' - Q_i \,.
        \end{equation}
        \item For $i \in \{ 0, \dots, n + 1 \}$, set $\delta_i = s_i \bar{W}_i'$.
        \item Set $\bar{W}_{n + 1} = 0$.
        Then, for $i$ from $n$ down to $0$, compute $\bar{W}_i$ using the recurrence relation
        \begin{equation}
            \label{eq:recurrence-W}
            \bar{W}_i = \bar{W}_{i + 1} + \delta_i - \delta_{i + 1} + P_i
            \,.
        \end{equation}
        \item For $i \in \{ 0, \dots, n + 1 \}$, set $W_i = \bar{W}_i / \bar{W}_0$, $W_i' = \bar{W}_i' / \bar{W}_0$, $W_i'' = \bar{W}_i'' / \bar{W}_0$.
        \item Build a piecewise linear interpolator $\widetilde{W}^{0}$ from $\{(r_i, W_i)\}_{i = 0}^{n + 1}$ for $W_{\bm{\theta}}$.
        \item Build a piecewise linear interpolator $\widetilde{W}^{1}$ from $\{(r_i, W_i')\}_{i = 0}^{n + 1}$ for $W_{\bm{\theta}}'$.
        \item Build a piecewise linear interpolator $\widetilde{W}^{2}$ from $\{(r_i, W_i'')\}_{i = 0}^{n + 1}$ for $W_{\bm{\theta}}''$.
    \end{enumerate}
\end{algorithm}

Now, we are ready to build a light version of $h_{\bm{\theta}}$.

\begin{algorithm}[Approximation of $h_{\bm{\theta}}$ from \gls{wt} estimates]
    \thlabel{alg:h-from-a}
    Let $\{x_i\}_{i = 0}^{k + 1}$ be the interpolation grid from \thref{alg:empirical-w-grid}.
    Let $\widetilde{W}^{0}$, $\widetilde{W}^{1}$ and $\widetilde{W}^{2}$ be the piecewise linear approximations to $W_{\bm{\theta}}$, $W_{\bm{\theta}}'$ and $W_{\bm{\theta}}''$ from \thref{alg:estimate-W}, respectively.
    To build an approximation for $h_{\bm{\theta}}$, follow these steps:
    \begin{enumerate}[noitemsep]
        \item For $i \in \{ 1, \dots, k \}$, set
        $W_i = \widetilde{W}^{0}(x_i)$, $W_i' = \widetilde{W}^{1}(x_i)$,
        $W_i'' = \widetilde{W}^{2}(x_i)$.
        \item Set $t_0 = 0$ and $t_{k + 1} = 1$.
        Then, for $i \in \{ 1, \dots, k \}$, set
        $t_i = (1 + x_i - W_i) / 2$.
        \item For $i \in \{ 1, \dots, k \}$, set
        $A_i = (1 + x_i + W_i) / 2$.
        \item For $i \in \{ 1, \dots, k \}$, set $M_i = 1 - W_i'$.
        \item For $i \in \{ 1, \dots, k \}$, set
        $A_i' = (1 + W_i') / M_i$.
        \item For $i \in \{ 1, \dots, k \}$, set
        $A_i'' = 4 W_i'' / M_i^3$.
        \item For $i \in \{ 1, \dots, k \}$, set $D_i = A_i' / A_i$.
        \item Set $h_0 = h_{k + 1} = 0$.
        Then, for $i \in \{ 1, \dots, k \}$, set
        \begin{equation}
            \label{eq:h-at-grid}
            h_i
            =
            1 + (1 - 2 t_i)
            D_i + t_i(1 - t_i)
            \left(
                \frac{A_i''}{A_i} -
                D_i^2
            \right) \,.
        \end{equation}
        \item Build a piecewise linear interpolator $\tilde{h}$ from $\{(t_i, h_i)\}_{i = 0}^{k + 1}$.
        \item Compute $I \approx \int_{0}^{1} \tilde{h}$ using the composite trapezoidal rule over $\{(t_i, h_i)\}_{i = 0}^{k + 1}$.
        \item Use $\hat{h} = \tilde{h} / I$ as an approximation for $h_{\theta}$ over $[0, 1]$.
    \end{enumerate}
\end{algorithm}

\thref{alg:h-from-a} deals with problems like the approximation of $h_{\bm{\theta}}$ and the rotation of $W_{\bm{\theta}}$.
On the other hand, \thref{alg:estimate-W} formalizes an efficient computation scheme for $W_{\bm{\theta}}$.
Both together allow computing $\ell(\bm{\theta} | \mathcal{D})$.
\figurename~\ref{fig:computation-graph} shows the full computation graph.
Gradients flow from the top \gls{pll} down the parameter vector using backpropagation.
We recommend using the \texttt{autograd} package~\cite{Maclaurin2015}, capable of performing automatic differentiation on native Python operations.
Some representative routines in our implementation are \texttt{numpy}'s \texttt{trapz}, for calculating integrals using the trapezoidal rule, and \texttt{cumsum}, for computing recurrences~\eqref{eq:recurrence-W-derivative} and~\eqref{eq:recurrence-W}.

\begin{figure}[ht]
    \centering
    \resizebox{\figurewidth}{!}{
        \input{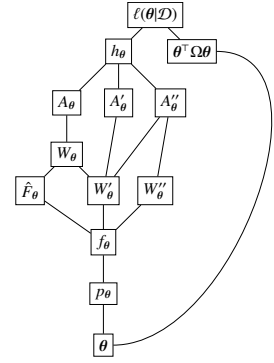}
    }
    \caption{
        Computation graph for the estimation process, from the bottom parameter vector $\bm{\theta}$ up to the \gls{pll} $\ell(\bm{\theta} | \mathcal{D})$.
    }
    \label{fig:computation-graph}
\end{figure}

\paragraph{Implementation tips}

\thref{alg:estimate-W} and \thref{alg:h-from-a} use discretization to approximate functions and integrals.
The finer-grained the discretization steps, the lower the error and the higher the computation time.
A trade-off between those dimensions is needed.
On the other hand, knowing $W_{\bm{\theta}}'(0^+) = -\infty$ and $W_{\bm{\theta}}''(0^+) = \infty$, we recommend choosing the grid in \thref{alg:estimate-W} so that points accumulate near zero, making the linear interpolation more effective.
Chebyshev nodes are a standard option.

To facilitate \gls{sbevc}'s estimation process, we propose to change the copula variable ordering whenever a steep slope is likely to appear for $W_{\bm{\theta}}$ near zero, which coincides with the minimum of $A_{\bm{\theta}}$ being placed at $t < 0.5$.
We can heuristically assess this situation by calculating the mode of the \gls{pdf} $h$, as suggested in~\cite{Eschenburg2013}.
If the mode appears at $t < 0.5$, the \gls{pf}'s minimum will \textit{likely} be placed at $t < 0.5$.
We support the hypothesis of~\cite{Eschenburg2013} based on our own experience.
Therefore, whenever the mode peaks at $t < 0.5$, we recommend changing the variable ordering before estimating and then flipping the resulting \gls{pf} $A$ as $\tilde{A}(t) = A(1 - t)$.

\subsection{Simulation}

Once the parameters $\bm{\theta}$ have been estimated, we propose to build $W_{\bm{\theta}}$ and $A_{\bm{\theta}}$ subsequently.
From that point on, querying the model (simulating, estimating probabilities, among others) will be equivalent to evaluating the \gls{pf} $A_{\bm{\theta}}$, as with any other \gls{evc}.

The algorithms in Section~\ref{subsec:estimation} stand valid, with some minor and convenient changes.
Since we only need to build the functions once, and not once per iteration, we may employ more expensive and accurate approximations.
In particular, $W_{\bm{\theta}}''$ can be evaluated without approximations.
More sophisticated procedures should replace trapezoidal rules and linear interpolations.
On the other hand, \thref{alg:empirical-w-grid} is no longer required.
Instead, we may employ a root-finding algorithm to invert the automorphism $t$.

The interpolation points of $A_{\bm{\theta}}$ could be input to the shape-preserving interpolation procedure by~\citeauthor{Schumaker1983}, which would guarantee that the resulting spline is convex over the whole domain~\cite{Schumaker1983}.
However, in general, the second derivative of such a spline would not be continuous, which would hinder the simulation process.
In practice, we recommend smoothness and accuracy over shape preservation, provided a sufficiently fine interpolation grid is used.

Finally, to draw samples from $C_{\bm{\theta}}$, we recommend the general algorithm in~\cite{Bouye2000}, which only requires inverting one of its partial derivatives~\cite{Doyon2013, Eschenburg2013}.
The root-finding algorithm in~\cite{Alefeld1995} proves to be highly effective.

\subsection{Refinement}

\label{subsec:refinement}

One of the limitations of \gls{sbevc} is the fact that an estimated \gls{pf} $A$ always satisfies $A'(0^+) = -1$ and $A'(1^-) = 1$.
These constraints are a consequence of our construction, which imposes $W'(0^+) = -\infty$ and $W'(1^-) = 0$ on the \gls{wt}.
In practice, however, these boundary constraints do not hinder the expressiveness of the resulting model.
Remember that, for instance, upper tail dependence does not relate to either boundary derivative of the \gls{pf}, but the mid-point value $A(1/2)$.
This fact contrasts with the nature of another semiparametric procedure like~\cite{HernandezLobato2011}, where a slope value entirely determined the tail index.

In \gls{sbevc}, misspecified slopes for the \gls{pf} have a much lower impact on the concordance (Blomqvist's beta) and upper tail dependence.
Nonetheless, since it might produce a slight bias, we propose a refinement step that could complement \gls{sbevc}.

\citeauthor{Khoudraji1995}'s method is best known for inducing asymmetry in symmetrical \glspl{evc}~\cite{Khoudraji1995}.
However, there is no reason why it could not apply to asymmetrical ones~\cite{Quessy2016}.
Consider a \gls{pf} $A$ obtained through \gls{sbevc}.
Differentiating~\eqref{eq:pickands-khoudraji}, we arrive at
\begin{equation}
    \label{eq:khoudraji-boundary-slopes}
    \begin{array}{r@{\ }c@{\ }l}
        A_{\alpha, \beta}'(0^+) &= \beta A'(0^+) &= -\beta \\
        A_{\alpha, \beta}'(1^-) &= \alpha A'(1^-) &= \alpha
    \end{array}
    \,,
\end{equation}
where, remember, $\alpha, \beta \in (0, 1]$, retrieving $A$ for $\alpha = \beta = 1$.
Even though $A$ is, in general, asymmetrical, we see from~\eqref{eq:khoudraji-boundary-slopes} that \citeauthor{Khoudraji1995}'s method serves our purpose of freely parameterizing the boundary slopes.\footnote{By convexity, the only \gls{evc} with either boundary slope equal to zero is the independence copula, with $A(t) = 1$ for all $t \in [0, 1]$. Therefore, except for this limiting case, both slopes are allowed to vary freely.}

We believe that adding two more parameters through \citeauthor{Khoudraji1995}'s method may improve the fitness of the resulting model in some particular cases, especially for weak correlations.
However, the inclusion of the new parameters in the gradient-based optimization seems unworkable, as it would invalidate the interpolation grid in~\thref{alg:h-from-a}.
A derivative-free optimization involving both the spline parameter vector $\bm{\theta}$ and the asymmetry parameters $\alpha$ and $\beta$ could be run,
starting from $\alpha = \beta = 1$ and some initial guess $\bm{\theta} = \bm{\theta}_0$ obtained through a gradient-based method.

    \section{Results}

\label{sec:results}

We will test \gls{sbevc} on simulated and actual data.
In both scenarios, we will compare \gls{sbevc} with the methodology by~\citeauthor{Cormier2014}~\cite{Cormier2014}.
We shall refer to their method as \gls{cobs}.
We have chosen \gls{cobs} for its flexibility and simplicity, sharing three relevant traits with \gls{sbevc}: complying with \gls{pf} constraints, using splines and exclusively addressing bivariate \glspl{evc}.

\subsection{Preliminaries}

\label{subsec:preliminaries}

Before diving into the specific settings of each experiment, let us clarify some shared configuration aspects.

\subsubsection{Optimization}

We performed all \gls{sbevc} estimates using standard Python scientific packages like \texttt{numpy} and \texttt{scipy}, and automatic differentiation, thanks to \texttt{autograd}~\cite{Maclaurin2015}.
Namely, we employed \texttt{scipy}'s implementation of the L-BFGS-B algorithm by~\citeauthor{Byrd1995}~\cite{Byrd1995}.
We assessed convergence by setting the \texttt{ftol=1e-6} configuration parameter in the \texttt{minimize} routine, which targets the relative change in the loss function between iterations.
Even though this value is very conservative, the procedure converges well, with reasonable execution times, as we will see.

All estimation runs started at the null spline, with all coordinates equal to zero, regardless of using a fixed affine centre.
Despite the caveats by~\citeauthor{HernandezLobato2011}~\cite{HernandezLobato2011}, as demonstrated in~\cite{Serrano2016}, current optimization methods can deal with complex problems even if the initial parameter values are far from the optimal solution.
Notwithstanding, we agree with~\citeauthor{HernandezLobato2011} that good initial guesses would speed up the process.

\gls{cobs} is very easy to implement on top of the \texttt{cobs} \textsf{R} package~\cite{Cormier2014}.
The \textsf{R} execution environment can be accessed from Python thanks to the \texttt{rpy2} Python package with little coding overhead.
Specifically, we employed the main \texttt{cobs} routine, selecting a smoothing splines regression of degree two by entering \texttt{lambda=-1}.
We raised the maximum number of iterations until convergence to 1,000 using the \texttt{maxiter} parameter.
We kept the maximum number of spline knots to the default value of 20.
Both the knots selection and the smoothing penalty was internally chosen by \texttt{cobs}.
The convexity requirement was introduced by setting \texttt{constraint="convex"}.
The boundary constraints were enforced over a fine 1,001-point equally-spaced grid over $[0, 1]$ using the \texttt{pointwise} argument.
Finally, we interpolated \texttt{cobs}' result over the former grid using cubic splines to allow for continuous second derivatives.

\subsubsection{Resources}

Both the \gls{ss} and actual data application would not have been possible without the vast repertoire of software artefacts and services currently available.

First, \gls{sbevc}, fully implemented in Python, was containerized using~\citetitle{Docker}~\cite{Docker}, which, apart from being ideal for achieving reproducible research, also helped to move our execution environment to the cloud with~\citetitle{CNCF}~\cite{CNCF}.
\citetitle{Docker} was also helpful for preparing a maintainable execution environment with Python and \textsf{R}, as required by \texttt{cobs}.

While developing and testing, we employed a local~\citetitle{MinikubeCommunity} cluster~\cite{MinikubeCommunity} on an Intel\textsuperscript{\tiny\textregistered} Core\textsuperscript{\tiny\textcopyright} i7-4700MQ CPU laptop with eight 2.40 GHz cores and 15.6 GiB of memory and operating system Ubuntu 20.04.3 LTS.
We entrusted the bulky \gls{ss} final executions to a cloud provider.
The Kubernetes service comprised 50 dynamically allocated nodes running on possibly different\footnote{Either (i) Platinum 8272CL, (ii) 8171M 2.1GHz, (iii) E5-2673 v4 2.3 GHz or (iv) E5-2673 v3 2.4 GHz.} Intel\textsuperscript{\tiny\textregistered} Xeon\textsuperscript{\tiny\textcopyright} architectures with Ubuntu 18.04.
Overall, each node counted on two virtual CPUs and seven GiB of memory at any given time.
Due to Kubernetes' requirements, only one CPU was available for Spark per node.

We sped up the experiments parallelizing specific tasks with~\citetitle{ASF2021}~\cite{ASF2021}.
To prepare the~\citetitle{ASF2021} setting with~\citetitle{CNCF}, two artefacts were of great help: the~\citetitle{DataMechanics}~\cite{DataMechanics} and the~\citetitle{GCP}~\cite{GCP}.

Finally,~\citetitle{ArgoProject}~\cite{ArgoProject} turned out to be helpful to manage all our~\citetitle{CNCF} experiments from the same friendly user interface, connecting the~\citetitle{CNCF} cluster to the Git repository.

\paragraph{Supplementary materials}

Access to source code and other deliverables will be provided upon acceptance for publication.

\subsection{Simulation study}

\label{subsec:simulations}

We conducted a \gls{ss} to test the effectiveness of \gls{sbevc} on a broad spectrum of cases with high confidence.
The \gls{ss} consists of three experiments.
The first one addresses the bias and variance tradeoffs by repeating the estimation process for many random samples drawn from a fixed copula in Table~\ref{tab:evc-families} for several parameter configurations.
The second experiment covers an even more extensive array of \glspl{evc} while focusing on validation through \gls{tvd}.
Then, the third one compares \gls{sbevc} with \gls{cobs} in terms of the \gls{rmise} in several scenarios with varying dependence strengths, asymmetries and sample sizes.

\gls{sbevc} entails numerous non-trivial analytic and geometric transformations.
Even though we could argue that none of them exceeds a reasonable level of complexity, clever algorithms and powerful computational resources are still needed for it to work in practice.
In particular, optimization algorithms are vital to finding solutions that maximize~\eqref{eq:loss-function} under memory and time constraints.
The joint behaviour of all these pieces is difficult to assess from a purely theoretical perspective without simulations.

\paragraph{Common settings}

Throughout the \gls{ss}, copula models build upon a cubic orthonormal \gls{zbs} basis.
The grid size was 200 both in \thref{alg:estimate-W} and \thref{alg:h-from-a} ($n + 2 = 200$ and $k + 2 = 200$, respectively).
Also, we took $\epsilon = 10^{-9}$ in \thref{alg:estimate-W}.
These settings express an adequate balance between approximation accuracy and reasonable execution times.

\subsubsection{Bias and variance}

The first part of the \gls{ss} consisted of 30 individual experiments, focusing on a particular instance of a copula family.
For each copula instance, we performed an estimation run with \gls{sbevc} on each of 100 different random samples from the copula for 3,000 runs.
Then, for each 100-sample experiment, we collected the pointwise means and pointwise 98\% confidence intervals of the estimated \glspl{pf} and compared both functional statistics with the original \gls{pf}.
We handled each random sample as an independent~\citetitle{ASF2021} task to speed up computations.

We employed the two families in Table~\ref{tab:evc-families}: the Gumbel and the Galambos families.
These are probably two of the most well-known \glspl{evc}.
\citeauthor{Genest2017} even studied and found a relation between the two in~\cite{Genest2017}, knowing the similarity of their \glspl{pf}.
In each copula family, we tested up to five different values of the unique parameter $\theta$, giving rise to different correlation levels.
Finally, apart from the pure form of each Gumbel or Galambos copula, we introduced asymmetry through~\citeauthor{Khoudraji1995}'s device, taking either $\alpha$ or $\beta$ equal to $0.5$ and leaving the other as $1$.
This configuration was precisely the one that demonstrated higher asymmetry in~\cite{Genest2011}.
As a side note, we remind that the asymmetrical extensions of the Gumbel and Galambos families are known as the Tawn and Joe families, respectively.

Each of the random copula samples consisted of 1,000 observations.
Our models were fit using 13 parameters in all cases: 10 more than the ground truth copula families.
We believe that the specific number of parameters has less impact in a semiparametric context, where one typically employs a large number and then reduces overfitting by penalizing curvature.
In the end, the target of this semiparametric method is a function that lives in an infinite-dimensional space.
In practice, both the Gumbel and the Galambos families need fewer parameters than 13, but in this case, we have preferred to stick to a large number to showcase the method's performance in a general setting.
Finally, for both the Gumbel and Galambos copulas, we used a curvature penalty factor $\lambda = 10^{-5}$.

\figurename~\ref{fig:simulation-study-gumbel} shows the results for the Gumbel copulas, while \figurename~\ref{fig:simulation-study-galambos} presents those of the Galambos.
The results are qualitatively very similar.
\gls{sbevc} displays low biases and variances in all cases.
If any, the highest biases appear under asymmetry and low correlations.
This behaviour matches the known limitation of \gls{sbevc} as regards the boundary slopes, which have fixed values.
Variance is also higher for small correlations, in agreement with~\cite{Kamnitui2019}.

\begin{figure*}
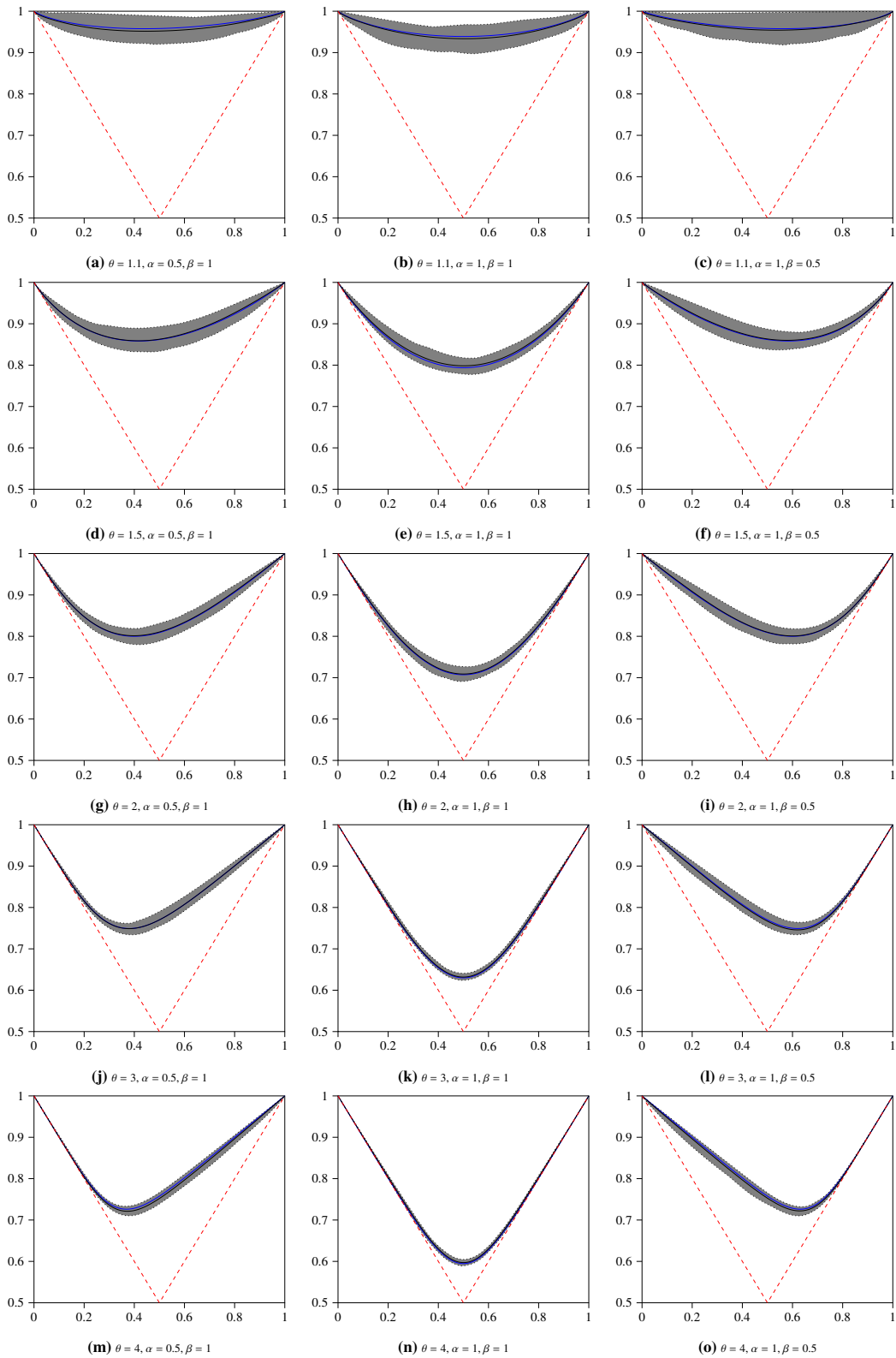

    \centering

    \begin{subfigure}[t]{\subfigurewidth}
        \centering
        \resizebox{\linewidth}{!} {
            \input{tikz/05-a-gumbel_1_1_alpha.tex}
        }
        \caption{\tiny $\theta = 1.1$, $\alpha = 0.5$, $\beta = 1$}
    \end{subfigure}
    \begin{subfigure}[t]{\subfigurewidth}
        \centering
        \resizebox{\linewidth}{!} {
            \input{tikz/05-b-gumbel_1_1.tex}
        }
        \caption{\tiny $\theta = 1.1$, $\alpha = 1$, $\beta = 1$}
    \end{subfigure}
    \begin{subfigure}[t]{\subfigurewidth}
        \centering
        \resizebox{\linewidth}{!} {
            \input{tikz/05-c-gumbel_1_1_beta.tex}
        }
        \caption{\tiny $\theta = 1.1$, $\alpha = 1$, $\beta = 0.5$}
    \end{subfigure}

    \begin{subfigure}[t]{\subfigurewidth}
        \centering
        \resizebox{\linewidth}{!} {
            \input{tikz/05-d-gumbel_1_5_alpha.tex}
        }
        \caption{\tiny $\theta = 1.5$, $\alpha = 0.5$, $\beta = 1$}
    \end{subfigure}
    \begin{subfigure}[t]{\subfigurewidth}
        \centering
        \resizebox{\linewidth}{!} {
            \input{tikz/05-e-gumbel_1_5.tex}
        }
        \caption{\tiny $\theta = 1.5$, $\alpha = 1$, $\beta = 1$}
    \end{subfigure}
    \begin{subfigure}[t]{\subfigurewidth}
        \centering
        \resizebox{\linewidth}{!} {
            \input{tikz/05-f-gumbel_1_5_beta.tex}
        }
        \caption{\tiny $\theta = 1.5$, $\alpha = 1$, $\beta = 0.5$}
    \end{subfigure}

    \begin{subfigure}[t]{\subfigurewidth}
        \centering
        \resizebox{\linewidth}{!} {
            \input{tikz/05-g-gumbel_2_alpha.tex}
        }
        \caption{\tiny $\theta = 2$, $\alpha = 0.5$, $\beta = 1$}
    \end{subfigure}
    \begin{subfigure}[t]{\subfigurewidth}
        \centering
        \resizebox{\linewidth}{!} {
            \input{tikz/05-h-gumbel_2.tex}
        }
        \caption{\tiny $\theta = 2$, $\alpha = 1$, $\beta = 1$}
    \end{subfigure}
    \begin{subfigure}[t]{\subfigurewidth}
        \centering
        \resizebox{\linewidth}{!} {
            \input{tikz/05-i-gumbel_2_beta.tex}
        }
        \caption{\tiny $\theta = 2$, $\alpha = 1$, $\beta = 0.5$}
    \end{subfigure}

    \begin{subfigure}[t]{\subfigurewidth}
        \centering
        \resizebox{\linewidth}{!} {
            \input{tikz/05-j-gumbel_3_alpha.tex}
        }
        \caption{\tiny $\theta = 3$, $\alpha = 0.5$, $\beta = 1$}
    \end{subfigure}
    \begin{subfigure}[t]{\subfigurewidth}
        \centering
        \resizebox{\linewidth}{!} {
            \input{tikz/05-k-gumbel_3.tex}
        }
        \caption{\tiny $\theta = 3$, $\alpha = 1$, $\beta = 1$}
    \end{subfigure}
    \begin{subfigure}[t]{\subfigurewidth}
        \centering
        \resizebox{\linewidth}{!} {
            \input{tikz/05-l-gumbel_3_beta.tex}
        }
        \caption{\tiny $\theta = 3$, $\alpha = 1$, $\beta = 0.5$}
    \end{subfigure}

    \begin{subfigure}[t]{\subfigurewidth}
        \centering
        \resizebox{\linewidth}{!} {
            \input{tikz/05-m-gumbel_4_alpha.tex}
        }
        \caption{\tiny $\theta = 4$, $\alpha = 0.5$, $\beta = 1$}
    \end{subfigure}
    \begin{subfigure}[t]{\subfigurewidth}
        \centering
        \resizebox{\linewidth}{!} {
            \input{tikz/05-n-gumbel_4.tex}
        }
        \caption{\tiny $\theta = 4$, $\alpha = 1$, $\beta = 1$}
    \end{subfigure}
    \begin{subfigure}[t]{\subfigurewidth}
        \centering
        \resizebox{\linewidth}{!} {
            \input{tikz/05-o-gumbel_4_beta.tex}
        }
        \caption{\tiny $\theta = 4$, $\alpha = 1$, $\beta = 0.5$}
    \end{subfigure}

    \caption{
        \gls{ss} for the Gumbel family.
        The blue line designates the ground-truth \gls{pf}, whereas the black corresponds to the estimations' pointwise mean.
        The shaded areas represent 98\% pointwise confidence intervals for the estimates.
    }
    \label{fig:simulation-study-gumbel}
\end{figure*}

\begin{figure*}
    \centering

    \begin{subfigure}[t]{\subfigurewidth}
        \centering
        \resizebox{\linewidth}{!} {
            \input{tikz/06-a-galambos_0_5_alpha.tex}
        }
        \caption{\tiny $\theta = 0.5$, $\alpha = 0.5$, $\beta = 1$}
    \end{subfigure}
    \begin{subfigure}[t]{\subfigurewidth}
        \centering
        \resizebox{\linewidth}{!} {
            \input{tikz/06-b-galambos_0_5.tex}
        }
        \caption{\tiny $\theta = 0.5$, $\alpha = 1$, $\beta = 1$}
    \end{subfigure}
    \begin{subfigure}[t]{\subfigurewidth}
        \centering
        \resizebox{\linewidth}{!} {
            \input{tikz/06-c-galambos_0_5_beta.tex}
        }
        \caption{\tiny $\theta = 0.5$, $\alpha = 1$, $\beta = 0.5$}
    \end{subfigure}

    \begin{subfigure}[t]{\subfigurewidth}
        \centering
        \resizebox{\linewidth}{!} {
            \input{tikz/06-d-galambos_0_75_alpha.tex}
        }
        \caption{\tiny $\theta = 0.75$, $\alpha = 0.5$, $\beta = 1$}
    \end{subfigure}
    \begin{subfigure}[t]{\subfigurewidth}
        \centering
        \resizebox{\linewidth}{!} {
            \input{tikz/06-e-galambos_0_75.tex}
        }
        \caption{\tiny $\theta = 0.75$, $\alpha = 1$, $\beta = 1$}
    \end{subfigure}
    \begin{subfigure}[t]{\subfigurewidth}
        \centering
        \resizebox{\linewidth}{!} {
            \input{tikz/06-f-galambos_0_75_beta.tex}
        }
        \caption{\tiny $\theta = 0.75$, $\alpha = 1$, $\beta = 0.5$}
    \end{subfigure}

    \begin{subfigure}[t]{\subfigurewidth}
        \centering
        \resizebox{\linewidth}{!} {
            \input{tikz/06-g-galambos_1_alpha.tex}
        }
        \caption{\tiny $\theta = 1$, $\alpha = 0.5$, $\beta = 1$}
    \end{subfigure}
    \begin{subfigure}[t]{\subfigurewidth}
        \centering
        \resizebox{\linewidth}{!} {
            \input{tikz/06-h-galambos_1.tex}
        }
        \caption{\tiny $\theta = 1$, $\alpha = 1$, $\beta = 1$}
    \end{subfigure}
    \begin{subfigure}[t]{\subfigurewidth}
        \centering
        \resizebox{\linewidth}{!} {
            \input{tikz/06-i-galambos_1_beta.tex}
        }
        \caption{\tiny $\theta = 1$, $\alpha = 1$, $\beta = 0.5$}
    \end{subfigure}

    \begin{subfigure}[t]{\subfigurewidth}
        \centering
        \resizebox{\linewidth}{!} {
            \input{tikz/06-j-galambos_1_5_alpha.tex}
        }
        \caption{\tiny $\theta = 1.5$, $\alpha = 0.5$, $\beta = 1$}
    \end{subfigure}
    \begin{subfigure}[t]{\subfigurewidth}
        \centering
        \resizebox{\linewidth}{!} {
            \input{tikz/06-k-galambos_1_5.tex}
        }
        \caption{\tiny $\theta = 1.5$, $\alpha = 1$, $\beta = 1$}
    \end{subfigure}
    \begin{subfigure}[t]{\subfigurewidth}
        \centering
        \resizebox{\linewidth}{!} {
            \input{tikz/06-l-galambos_1_5_beta.tex}
        }
        \caption{\tiny $\theta = 1.5$, $\alpha = 1$, $\beta = 0.5$}
    \end{subfigure}

    \begin{subfigure}[t]{\subfigurewidth}
        \centering
        \resizebox{\linewidth}{!} {
            \input{tikz/06-m-galambos_3_alpha.tex}
        }
        \caption{\tiny $\theta = 3$, $\alpha = 0.5$, $\beta = 1$}
    \end{subfigure}
    \begin{subfigure}[t]{\subfigurewidth}
        \centering
        \resizebox{\linewidth}{!} {
            \input{tikz/06-n-galambos_3.tex}
        }
        \caption{\tiny $\theta = 3$, $\alpha = 1$, $\beta = 1$}
    \end{subfigure}
    \begin{subfigure}[t]{\subfigurewidth}
        \centering
        \resizebox{\linewidth}{!} {
            \input{tikz/06-o-galambos_3_beta.tex}
        }
        \caption{\tiny $\theta = 3$, $\alpha = 1$, $\beta = 0.5$}
    \end{subfigure}

    \caption{
        \gls{ss} for the Galambos family.
        The blue line designates the ground-truth \gls{pf}, whereas the black corresponds to the estimations' pointwise mean.
        The shaded areas represent 98\% pointwise confidence intervals for the estimates.
    }
    \label{fig:simulation-study-galambos}
\end{figure*}

In all simulations, we employed the trick mentioned in Section~\ref{subsec:estimation} for selecting the \textit{a priori} more convenient variable ordering to avoid numerical instabilities.
The procedure worked well, as demonstrated by the nearly identical results obtained for either $\alpha = 0.5$ or $\beta = 0.5$.

\paragraph{Execution details}

We approximately recorded the execution times for the 30 experiments with the assistance of the Spark Web UI.
Roughly 50\% finished in three minutes, 75\% in four and 90\% in seven.
Consistently below eight minutes, the most time-consuming experiments correspond to the highest correlated Gumbel and Galambos copulas.
That is presumably because the optimal solution was furthest from the starting null vector.
Since there were 100 tasks on each job and the cluster only had 50 nodes, we could expect the average execution time to be half of the previous values.

\subsubsection{Total variation}

\label{subsubsec:simulation-study-total-variation}

In the second part of the \gls{ss}, we generated $n = 200$ random \glspl{sbevc}.
We chose the affine spline model with centre~\eqref{eq:orthogonal-projection} as the first building block, assuming uniformly distributed knots and $d = 13$ parameters.
Let us call $\bm{\theta}_0 \in \mathbb{R}^d$ the coordinates of the centre of the affine model.
We then ran an MCMC simulation assuming the model coordinates $\bm{\theta}$ in~\eqref{eq:affine-bayes-model} were distributed according to the following \gls{pdf}:
\begin{equation*}
    p({\bm{\theta}}) \propto
    \left\{
        \begin{aligned}
            e^{-\lambda {\bar{\bm{\theta}}}^{\top} \Omega \bar{\bm{\theta}}} \,,
                & \ \text{if} \ \lVert \bm{\theta} \rVert_2 \leq R \\
            0 \,,
                & \ \text{if} \ \lVert \bm{\theta} \rVert_2 > R
        \end{aligned}
    \right.
    \,,
    \ \text{for} \
    \bar{\bm{\theta}} = \bm{\theta} + \bm{\theta}_0
    \,,
\end{equation*}
where $\Omega$ is the curvature matrix of the underlying spline, as described in Section~\ref{subsec:estimation}, and $\lambda$ and $R$ are tuning parameters.
The previous model is the truncated version of an improper prior based on curvature penalization, with factor $\lambda$.
The support of the distribution is the hyperball of radius $R$.

We tuned the parameters with values $\lambda = 10^{-4}$ and $R = 5$ so that the resulting splines covered a wide range of correlations (in the sense of the \gls{gc}) and were, at the same time, smooth.
Finally, to prevent any asymmetry, we replaced the even elements in the sequence with their corresponding mirrored versions $\tilde{A}(t) = A(1 - t)$.
\figurename~\ref{fig:random-pickands} shows a subsample of the generated random \glspl{pf}.
They cover the area between the support lines and the upper bound line in a reasonably balanced way.
We also employed the heuristic to determine the most suitable variable ordering in this part of the \gls{ss}.
Hence, we expected \gls{sbevc} to perform well regardless of the orientation of the \gls{pf}.

\begin{figure}[ht]
    \centering
    \resizebox{\figurewidth}{!}{
        \input{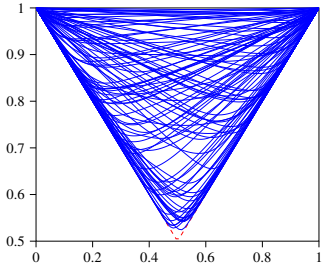}
    }
    \caption{
        A subsample of size 100 from the whole population of random \glspl{pf} used in the second part of the \gls{ss}.
    }
    \label{fig:random-pickands}
\end{figure}

For each element in the sequence $\{\bm{\theta}_i\}_{i = 1}^n$, we built the \gls{evc} $C_{\bm{\theta}_i}$ and performed several estimation runs on random samples of different sizes $\{S_j\}_{j = 1}^m$.
All fitted models had the same number of parameters as the ground truth splines ($d = 13$) and employed the same affine translation.
The penalty factor in the loss function~\eqref{eq:loss-function} was also set to $\lambda = 10^{-4}$.
The only aspect in which estimated models differed from ground truth is spline knot placement, which was uniform for the latter, but empirically assessed for the former.
Then, for each sample size $S_j$, we estimated a copula $C_{ij}$ using \gls{sbevc} and assessed divergence from ground truth through \gls{tvd}
\begin{equation}
    \label{eq:total-variation-copulas}
    d_{\text{TV}}(C_{\bm{\theta}_i}, C_{ij}) =
    \frac{1}{2}
    \int_{[0, 1]^2} \lvert c_{\bm{\theta}_i} (u, v) - c_{ij} (u, v) \rvert
    \ du dv
    \,,
\end{equation}
where $c_{\bm{\theta}_i}$ and $c_{ij}$ are the \glspl{pdf} of $C_{\bm{\theta}_i}$ and $C_{ij}$, respectively.
The \gls{tvd} defined in~\eqref{eq:total-variation-copulas} is the bivariate counterpart of \thref{def:total-variation} and thus provides an upper bound on the difference between the measured values of each copula on any measurable set $B \subset [0, 1]^2$.
Therefore,~\eqref{eq:total-variation-copulas} is a very conservative evaluation measure.

Table~\ref{tab:random-evcopulas} presents the main summary statistics from the experiment.
Each~\citetitle{ASF2021} task targeted a different random \gls{evc}.
Then, each task comprised four estimation runs.
The table shows promising results, considering the complexity of the ground truth models and the finiteness of samples.
Mean values are typically below $0.05$, whereas the 75\% and 90\% quantiles do not surpass the $0.10$ threshold.
Table~\ref{tab:random-evcopulas} shows that \gls{tvd} decreases as the sample size increases.
\figurename~\ref{fig:total-variation-boxplot} reveals some outliers, which become rarer with larger sample sizes.
Besides, \figurename~\ref{fig:total-variation-boxplot} suggests \gls{sbevc} can handle even the highest correlated samples well.
The outliers are due to the sensitivity of the \gls{tvd} metric to deviations in highly correlated samples.
As \figurename~\ref{fig:total-variation-outliers} shows, the \gls{tvd} metric positively correlates with the \gls{gc} even for moderate values of the latter.

\begin{table}
    \centering
    \resizebox{\columnwidth}{!} {
        \begin{tabular}{lllllll}
\toprule
{} &    mean &     10\% &     25\% &     50\% &     75\% &     90\% \\
sample size &         &         &         &         &         &         \\
\midrule
250         &  .06642 &   .0226 &  .03237 &  .04578 &  .06279 &  .08914 \\
500         &  .04591 &  .01562 &  .02492 &  .03797 &  .05527 &  .07525 \\
1000        &  .03701 &  .01333 &   .0213 &  .03093 &  .04301 &  .05994 \\
2000        &   .0322 &  .01157 &  .01681 &  .02518 &  .03421 &  .05841 \\
\bottomrule
\end{tabular}

    }
    \caption{
        Main summary statistics from the second part of the \gls{ss}, mean and quantiles, for each sample size.
        \gls{tvd} decreases as sample size increases.
    }
    \label{tab:random-evcopulas}
\end{table}

\begin{figure}[ht]
    \centering
    \resizebox{\figurewidth}{!}{
        \input{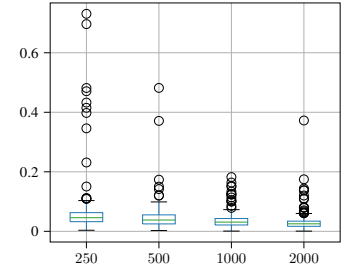}
    }
    \caption{
        Box plots of the \gls{tvd} distributions for each sample size.
    }
    \label{fig:total-variation-boxplot}
\end{figure}

\begin{figure}[ht]
    \centering
    \resizebox{\columnwidth}{!}{
        \includegraphics[width=\columnwidth]{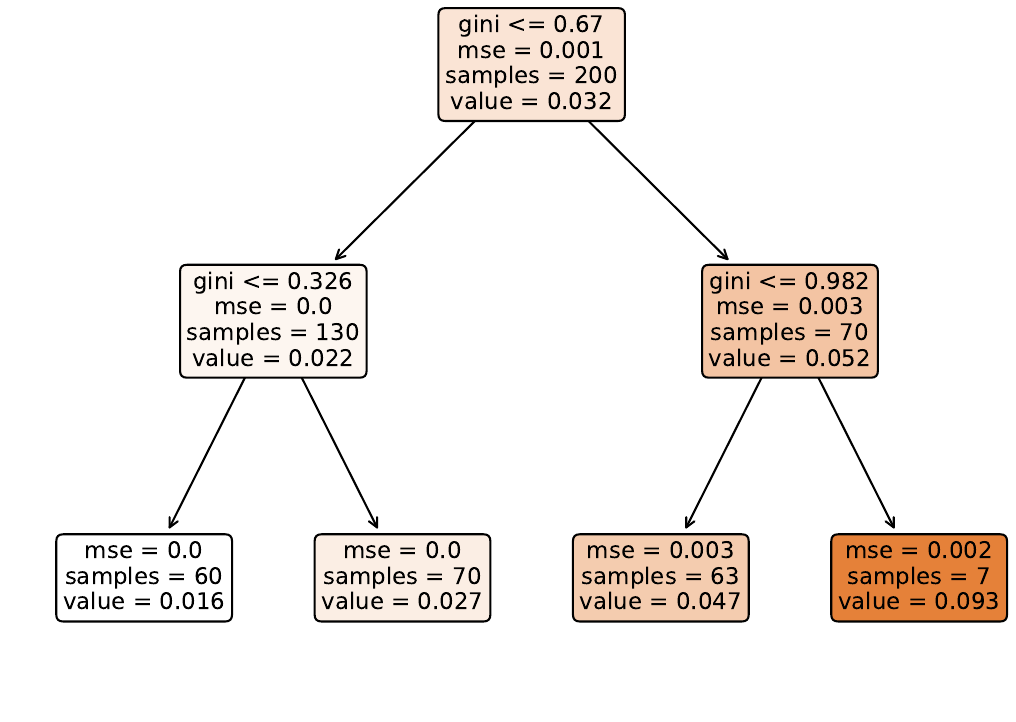}
    }
    \caption{
        Regression tree (depth 2) of \gls{tvd} on the \gls{gc} for sample size 2,000.
        On average, \gls{sbevc} performs well even on the seven instances with \gls{gc} above 98\%.
    }
    \label{fig:total-variation-outliers}
\end{figure}

\paragraph{Execution details}

The total execution time for the 200 random \glspl{evc} was roughly one hour.
Considering the cluster only had 50 nodes, the average execution time for a task in this experiment was approximately 15 minutes.

\subsubsection{RMISE}

The third part of the \gls{ss} consisted of 20 individual experiments with different settings.
We generated 100 random samples in each experiment and fitted them using \gls{sbevc} and \gls{cobs}.
Then, we collected the squared $L^2([0, 1])$ distances between the ground-truth \gls{pf} and the estimated \gls{pf} through \gls{sbevc} and \gls{cobs}.
For each estimation method, these measures were averaged to approximate the \gls{rmise} like in~\cite{Vettori2017}.
The lower the \gls{rmise}, the better the technique.
We also assessed the statistical significance of the results through a Wilcoxon signed-rank test applied on the unaggregated squared $L^2([0, 1])$ distances, assuming the null hypothesis that \gls{cobs} produces better results than \gls{sbevc}.
As in the previous parts of the \gls{ss}, we used for \gls{sbevc} $d = 13$ \gls{zbs} elements and the curvature penalty factor was $\lambda = 10^{-4}$.

\begin{table}
    \centering
    \resizebox{\columnwidth}{!} {
        \begin{tabular}{lrrr}
\toprule
{} & $\mathrm{RMISE}(\mathrm{SBEVC})$ & $\mathrm{RMISE}(\mathrm{COBS})$ &   p-value \\
settings                                  &                                  &                                 &           \\
\midrule
$\theta = 1.1$, $\beta = 0.5$, $n = 250$  &                           .02375 &                          .04314 &  2.29E-08 \\
$\theta = 1.1$, $\beta = 0.5$, $n = 1000$ &                           .01441 &                          .02611 &  4.57E-10 \\
$\theta = 1.1$, $\beta = 1.0$, $n = 250$  &                           .02497 &                          .04534 &  7.80E-07 \\
$\theta = 1.1$, $\beta = 1.0$, $n = 1000$ &                           .01283 &                          .02362 &  2.60E-10 \\
$\theta = 1.5$, $\beta = 0.5$, $n = 250$  &                           .02097 &                          .04466 &  1.15E-13 \\
$\theta = 1.5$, $\beta = 0.5$, $n = 1000$ &                           .01198 &                          .02749 &  9.90E-16 \\
$\theta = 1.5$, $\beta = 1.0$, $n = 250$  &                           .01518 &                          .04018 &  5.15E-14 \\
$\theta = 1.5$, $\beta = 1.0$, $n = 1000$ &                           .00807 &                          .02326 &  4.64E-14 \\
$\theta = 2.0$, $\beta = 0.5$, $n = 250$  &                           .01558 &                          .04233 &  5.86E-15 \\
$\theta = 2.0$, $\beta = 0.5$, $n = 1000$ &                           .00839 &                          .02384 &  1.14E-15 \\
$\theta = 2.0$, $\beta = 1.0$, $n = 250$  &                           .01116 &                          .02854 &  1.15E-13 \\
$\theta = 2.0$, $\beta = 1.0$, $n = 1000$ &                           .00527 &                          .01835 &  6.94E-17 \\
$\theta = 3.0$, $\beta = 0.5$, $n = 250$  &                           .01276 &                          .03909 &  1.21E-17 \\
$\theta = 3.0$, $\beta = 0.5$, $n = 1000$ &                           .00626 &                          .02494 &  5.20E-17 \\
$\theta = 3.0$, $\beta = 1.0$, $n = 250$  &                           .00721 &                          .04049 &  2.82E-14 \\
$\theta = 3.0$, $\beta = 1.0$, $n = 1000$ &                           .00490 &                          .01568 &  1.34E-15 \\
$\theta = 4.0$, $\beta = 0.5$, $n = 250$  &                           .01128 &                          .03556 &  7.79E-17 \\
$\theta = 4.0$, $\beta = 0.5$, $n = 1000$ &                           .00487 &                          .02241 &  5.74E-18 \\
$\theta = 4.0$, $\beta = 1.0$, $n = 250$  &                           .00986 &                          .05981 &  1.92E-15 \\
$\theta = 4.0$, $\beta = 1.0$, $n = 1000$ &                           .00513 &                          .01135 &  7.04E-08 \\
\bottomrule
\end{tabular}

    }
    \caption{
        A comparison between \gls{sbevc} and \gls{cobs} based on \gls{rmise}.
        All settings refer to the parametric Gumbel family in Table~\ref{tab:evc-families}.
        \gls{rmise} is statistically significantly lower for \gls{sbevc}.
    }
    \label{tab:rmise-cormier}
\end{table}

The \gls{rmise} results are presented in Table~\ref{tab:rmise-cormier}.
Each row corresponds to a different experiment.
The left-most column shows the experiment settings.
The underlying parametric copula belongs to an asymmetrical Gumbel family (see Table~\ref{tab:evc-families}) using~\citeauthor{Khoudraji1995}'s device~\eqref{eq:pickands-khoudraji}.
Then, $\theta$ is the parameter of the Gumbel copula and $\beta$ is one of the asymmetry parameters, fixing $\alpha = 1$ as a constant throughout all configurations.
The combinations of $\theta$, and $\beta$ are precisely those that appear in the mid and right columns in \figurename~\ref{fig:simulation-study-gumbel}.
We draw samples from each copula configuration with a medium ($n = 1,000$) and a small ($n = 250$) sample size.
As we can see, \gls{sbevc} significantly outperforms \gls{cobs} in all circumstances, roughly halving the \gls{rmise}.
\subsection{Case study}

\label{subsec:applications}

The following sections will solve a statistical modelling and simulation problem on LIGO and Virgo's precious \gls{gw} detection data.
The aim of this case study is twofold.
On the one hand, we will examine the steps in the construction and estimation of \glspl{sbevc} with an authentic hands-on experience.
On the other hand, we aim to compare \gls{sbevc} with \gls{cobs} on non-synthetic samples.
After some sensible transformations, LIGO and Virgo's data have an \gls{evc} dependence structure.
The underlying copula presents a very different look than what we have seen in \figurename~\ref{fig:simulation-study-gumbel} and \figurename~\ref{fig:simulation-study-galambos}.

\subsubsection{History}

In 2015, the LIGO\footnote{Laser Interferometer Gravitational-Wave Observatory.}
Scientific Collaboration and the Virgo Collaboration announced the first direct detection of a \gls{gw}, produced by the merger of a binary black hole~\cite{Abbott2016}.
The existence of \glspl{gw}, ripples in the fabric of space-time, was predicted by Einstein's theory of general relativity in 1916 as a mathematical construct that many thought to have no physical meaning~\cite{CervantesCota2016}.
It took nearly a century from its prediction and 60 years of search to experimentally ascertain the discovery, opening a new era for astronomy.

Only the most extreme events in the Universe, in terms of energy, can generate \glspl{gw} strong enough to be detected by current
experimental procedures due to the small value of the gravitational constant~\cite{CervantesCota2016}, which expresses the rigidity of space-time.
A significant amount of human and material resources are needed to detect \glspl{gw}.
Specifically, sufficiently sensitive interferometers need to have arms several kilometres long.
Additionally, in order to discriminate between true detections and spurious local signals (like electromagnetic radiation or earthquakes), several detectors, far apart from each other, are needed.

LIGO, settled in the United States, with two laboratories, was the first detector of an advanced global network that aims to increase discoveries' accuracy and exhaustiveness~\cite{Abbott2016}, soon to be joined by others, most notably Virgo, in Italy.
Despite LIGO and Virgo joining efforts, it was LIGO that reported the first detection since the Virgo facilities were not operating at that time for upgrading reasons.
Since the first detection in 2015, the collaboration of LIGO and Virgo has confirmed 50 events.
They all correspond to massive body mergers, mainly black holes and neutron stars.

\subsubsection{Data}

We have chosen the \gls{gw} detection dataset gathered by the LIGO and Virgo collaborations during their first three observation runs to test the applicability of \gls{sbevc}.
It consists of 50 rows and two columns.
Each row represents a merger event, while each column features one of the masses involved in the event, measured in solar mass units (M\textsubscript{\(\odot\)}).
During the first and second observation runs, 11 events were detected, while the third run provided 39.
The first event was GW150914, in September 2015, and the last one, GW190930\_133541, in September 2019.
LIGO and Virgo report the larger of the two masses, the primary mass, as the first tuple component, followed by the secondary mass.

We believe that very few datasets better represent bivariate data, considering the very nature of \glspl{bm}.
Bivariate models are usually building blocks for higher-dimensional ones, but in this case, all the attention is focused on two mass values of high scientific relevance.
Another aspect that adds to this significance is the scarcity of data, for only 50 events have been recorded during five years.
This scarcity contrasts with the increasingly large amounts of information coming from IoT, social networks, finance, among others, in the current era of Big Data.

\subsubsection{Model}

As mentioned above, the dataset consists of 50 bivariate observations $\mathcal{D} = \{(M_1^{(i)}, M_2^{(i)})\}_{i = 1}^{50}$, where $M_1^{(i)} \geq M_2^{(i)}$.
The last censoring constraint makes the dataset not directly tractable by usual copulas, supported on the whole $[0, 1]^2$, unless conveniently preprocessed.

LIGO and Virgo perform a statistical analysis of the joint mass distribution~\cite{Abbott2021}.
They consider two separate univariate models.
The first one models the primary mass $M_1$ unchanged, whereas the second one models the \textit{mass ratio} $Q = M_2 / M_1$ conditioning on $M_1$.
Since $M_1 \geq M_2$, by definition, the resulting model captures by construction the censoring constraint.
The final joint model is formed by the vector $(M_1, Q M_1)$.

Instead of considering an auxiliary ratio variable, we directly model a bivariate mass vector.
We turned the censoring problem into an exchangeable one, where both masses played the same role.
The original dataset $\mathcal{D}$ does not allow such a treatment, so we hypothesized a new sample space where primary masses are detected with 50\% probability at the first vector component and 50\% at the second one.
This scenario corresponds to detections reporting masses without considering their relative order.
Therefore, we built a new sample $\widetilde{\mathcal{D}} = \{(\widetilde{M}_1^{(i)}, \widetilde{M}_2^{(i)})\}_{i = 1}^{100}$, where $\widetilde{M}_j^{(i)} = M_j^{(i)}$ or $\widetilde{M}_j^{(i)} = M_{1 + j \bmod 2}^{(i - 50)}$, respectively, if $i \leq 50$ or $i > 50$.
We then targeted a random vector $(\widetilde{M}_1, \widetilde{M}_2)$.
To retrieve the original primary-secondary mass model, we just had to take $M_1 = \max \{\widetilde{M}_1, \widetilde{M}_2\}$ and $M_2 = \min \{\widetilde{M}_1, \widetilde{M}_2\}$.

Using the previous up-sampled and symmetrical dataset, we fitted (i) a single univariate mass model $f$ for both margins and (ii) a copula model $C$ of the dependency between mass ranks.

\paragraph{Univariate margin model}

We decided to employ a semiparametric model for the univariate margin mass model.
We successfully tried the same technique we used for modelling the density $f$ in~\eqref{eq:williamson}: Bayes space \glspl{pdf} built from \glspl{zbs}.

The result of our experiment is shown in \figurename~\ref{fig:ligo-pdf}.
We selected 17 parameters, with knots distributed according to the original sample between 1 M\textsubscript{\(\odot\)} and 100 M\textsubscript{\(\odot\)}, and a curvature penalty factor of 10.
The first mode, near 1 M\textsubscript{\(\odot\)}, mostly corresponds to neutron stars; black hole masses typically range beyond 5 M\textsubscript{\(\odot\)}.

\begin{figure}[ht]
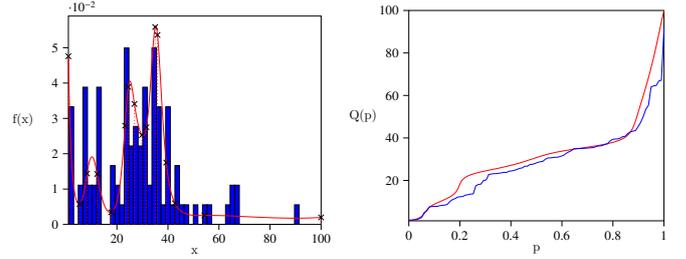

    \centering
    \begin{subfigure}{\figurewidth}
        \centering
        \resizebox{\columnwidth}{!} {
            \input{tikz/10-a-ligo-pdf.tex}
        }
    \end{subfigure}%
    \begin{subfigure}{\figurewidth}
        \centering
        \resizebox{\columnwidth}{!} {
            \input{tikz/10-b-ligo-inv-cdf.tex}
        }
    \end{subfigure}
    \caption{
        Univariate margin mass model.
        The \gls{pdf} is displayed on the left, whereas the quantile function is on the right.
        In both cases, the fitted model shows in red and the empirical estimate is in blue.
        The vertical cuts on the left correspond to the underlying spline knots.
    }
    \label{fig:ligo-pdf}
\end{figure}

\paragraph{Bivariate copula model}

Letting $\hat{F}$ be the empirical \gls{cdf} of the univariate sample $\{\widetilde{M}_1^{(i)}\}_{i = 1}^{100}$ (equivalently, from $\{\widetilde{M}_2^{(i)}\}_{i = 1}^{100}$), we decided to fit a copula pseudo-sample
$
    \widetilde{\mathcal{D}}_{\text{cop}} =
    \{( \widetilde{U}_1^{(i)}, \widetilde{U}_2^{(i)}) \}_{i = 1}^{100} =
    \{
        (
            \hat{F}(\widetilde{M}_1^{(i)}),
            \hat{F}(\widetilde{M}_2^{(i)})
        )
    \}_{i = 1}^{100}
$
independent of the fitted margin model from the previous section.

The applicability of \glspl{evc} was readily made clear after inspecting $\widetilde{\mathcal{D}}_{\text{cop}}$, where the mirrored data points resembled some characteristic patterns we saw during a random \gls{evc} generation run \textit{\`{a} la} \figurename~\ref{fig:random-pickands}.
Namely, they outlined two curved paths that met at both the lower and upper tail corners.

Data inspection also revealed the absence of upper tail dependence, while lower tail dependence was present.
This behaviour did not match the features of \glspl{evc}: in practice, they never have lower tail dependence, but they do exhibit dependence in the upper tail.
Interestingly, we can resort to \textit{survival copulas} whenever a switch between lower and upper tails is needed~\cite{Eschenburg2013}.
If a random vector with uniform margins $(U, V)$ is distributed according to a copula $C$, then $(1 - U, 1 - V)$ follows the survival copula~\cite{Bouye2000, Georges2001} $\check{C}(u, v) = u + v - 1 + C(1 - u, 1 - v)$.
The bivariate copula sample $\widetilde{\mathcal{D}}_{\text{cop}}$ was accordingly transformed into
$
    \widetilde{\mathcal{D}}_{\text{surv}} =
    \{
        (
            1 - \widetilde{U}_1^{(i)},
            1 - \widetilde{U}_2^{(i)})
        )
    \}_{i = 1}^{100}
$.
Once $\check{C}$ fits $\widetilde{\mathcal{D}}_{\text{surv}}$, the original copula can be retrieved by taking $C$ equal to the survival copula of $\check{C}$.

An extreme-value dependence test~\cite{Ghorbal2009}, implemented in the function \texttt{evTestK} of the \textsf{R} package \texttt{copula}~\cite{Hofert2020, IvanKojadinovic2010, JunYan2007, MariusHofert2011}, confirmed our intuition about the applicability of \glspl{evc}, yielding a p-value higher than 0.35.

The \gls{sbevc} model builds upon a cubic (orthonormal) \gls{zbs} basis with 13 elements, a curvature penalty factor of $10^{-5}$ and interpolation grid sizes of $k + 2 = 80$, in \thref{alg:h-from-a}, and $n + 2 = 200$, in \thref{alg:estimate-W}.
The value of the latter setting is lower than the one employed in the \gls{ss} based on the reduced sample size.
On the other hand, we used the same \gls{cobs} configuration as in Section~\ref{subsec:simulations}.

\figurename~\ref{fig:ligo-zfw} shows the final state of \gls{sbevc}'s internal functions defined in the \gls{wt} domain.
The resulting Bayes density has two main modes, yielding a \gls{wt} with a linear region.
On the other hand, \figurename~\ref{fig:ligo-ah} shows the estimated \gls{pf} and its correspondent $h$ density~\eqref{eq:h-pdf}.
Despite the sample $\widetilde{\mathcal{D}}_{\text{surv}}$ being exchangeable, the $h$ estimate fails to be perfectly symmetrical, with the left peak a bit higher than the one on the right.
This behaviour was not wholly unexpected, given that \gls{sbevc} does not address symmetry specifically.
Taking that into account, \figurename~\ref{fig:ligo-ah} shows that symmetry is reasonably well captured.
Notwithstanding, before reversing the survival model, we decided to apply a symmetrization procedure on the resulting \gls{pf} $A$, considering $\tilde{A}(t) = [A(t) + A(1 - t)] / 2$.

\begin{figure}[ht]
    \center
    \resizebox{\figurewidth}{!}{
        \input{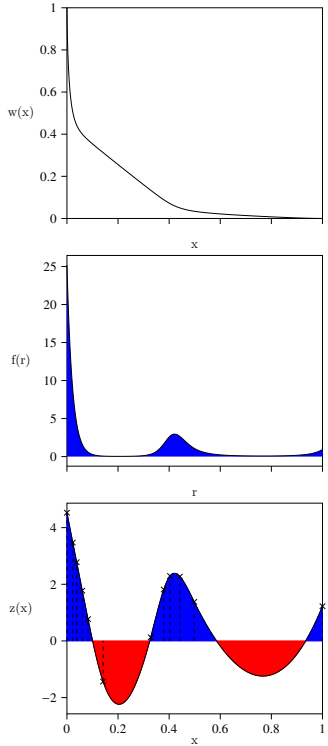}
    }
    \caption{
        Internal function constructs $z$ (zero-integral spline), $f$ (Bayes density) and $W$ (\gls{wt}).
        The plot displays the underlying spline knots of $z$.
    }
    \label{fig:ligo-zfw}
\end{figure}

\begin{figure}[ht]
    \center
    \resizebox{\figurewidth}{!}{
        \input{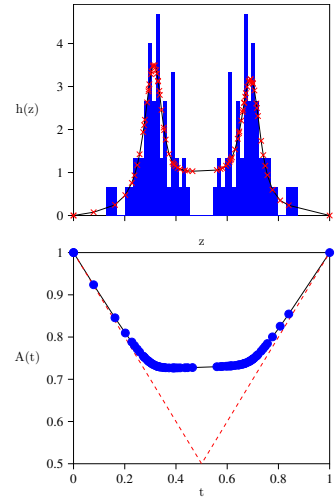}
    }
    \caption{
        \gls{pf} $A$ and target density $h$.
        The knots represent the function values at the $t_i$'s grid defined in \thref{alg:h-from-a}.
    }
    \label{fig:ligo-ah}
\end{figure}

\figurename~\ref{fig:ligo-bayes} shows the Bayesian posterior distribution of \gls{sbevc} \glspl{pf}, using the previous \gls{pll} result as an initial guess for the MCMC sampling.
We drew a million random observations from the posterior distribution.
The job was divided into 100~\citetitle{ASF2021} tasks corresponding to MCMC runs with 100 independent walkers~\cite{ForemanMackey2013}, each one generating 200 observations, with a \textit{burn-in} period of 100.
The confidence interval turned out to be wider than expected but preserving the overall shape.
In turn, \figurename~\ref{fig:ligo-pickands-cobs} shows the \gls{cobs} model, which happens to have a very different shape from \figurename~\ref{fig:ligo-bayes}, lacking a flat central region.
\figurename~\ref{fig:ligo-pickands-cobs} demonstrates that \gls{cobs} captures symmetry well.
Interestingly, the \gls{cobs} model falls outside the confidence interval in \figurename~\ref{fig:ligo-bayes}, indicating that \gls{sbevc} and \gls{cobs} have very different approaches to data fitting.

\begin{figure}[ht]
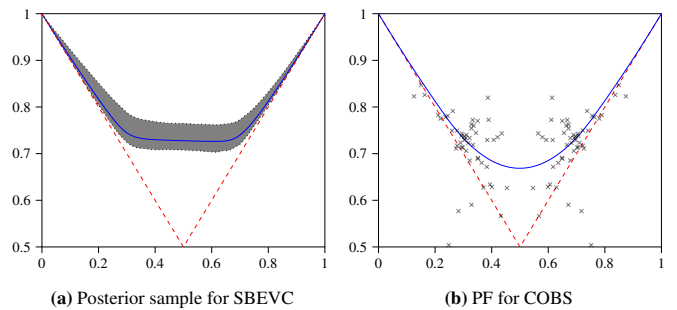

    \center
    \begin{subfigure}{\figurewidth}
        \centering
        \resizebox{\columnwidth}{!}{
            \input{tikz/13-a-ligo-bayes.tex}
        }
        \caption{Posterior sample for \gls{sbevc}}
        \label{fig:ligo-bayes}
    \end{subfigure}%
    \begin{subfigure}{\figurewidth}
        \centering
        \resizebox{\columnwidth}{!}{
            \input{tikz/13-b-ligo-cormier-pickands.tex}
        }
        \caption{\gls{pf} for \gls{cobs}}
        \label{fig:ligo-pickands-cobs}
    \end{subfigure}
    \caption{
        On the left, pointwise confidence interval (98\%) and mean of the posterior \glspl{pf} sample from the MCMC simulation.
        On the right, the \gls{cobs} \gls{pf} model fitting the empirical graph $(Z_i, T_i)$.
    }
    \label{fig:ligo-bayes-pickands-cobs}
\end{figure}

\figurename~\ref{fig:ligo-copulas} shows the corresponding sample-density plots for \gls{sbevc} and \gls{cobs} after reversal of the survival transformation.
The \glspl{pdf} capture the presence of lower tail dependence and the absence of upper tail dependence in both cases.
The correlation is also very similar.
However, there is a remarkable density gap in \figurename~\ref{fig:ligo-sbevc-copula} in the region surrounding the diagonal $\{u = v\}$ that is not present in \figurename~\ref{fig:ligo-cobs-copula}.
This is how the presence or absence in \figurename~\ref{fig:ligo-bayes-pickands-cobs} of a flat region translates to \glspl{pdf}.
Consequently, \gls{sbevc} and \gls{cobs} disagree when evaluating the chances of \glspl{bm} involving similar masses.
However, the fitted observations from $\widetilde{\mathcal{D}}_{\text{cop}}$ seem to better support \gls{sbevc}'s hypothesis than \gls{cobs}'.

\begin{figure}[ht]
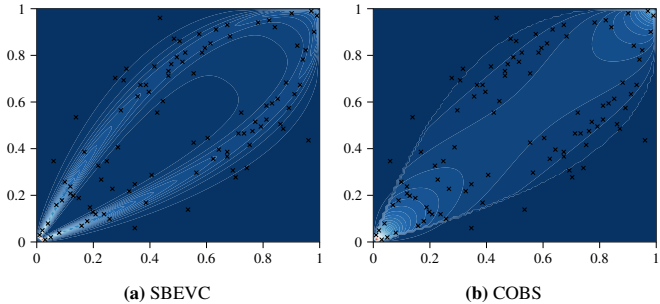

    \center
    \begin{subfigure}{\figurewidth}
        \centering
        \resizebox{\columnwidth}{!}{
            \input{tikz/14-a-ligo-copula.tex}
        }
        \caption{
            \gls{sbevc}
        }
        \label{fig:ligo-sbevc-copula}
    \end{subfigure}%
    \begin{subfigure}{\figurewidth}
        \centering
        \resizebox{\columnwidth}{!}{
            \input{tikz/14-b-ligo-cormier-copula.tex}
        }
        \caption{
            \gls{cobs}
        }
        \label{fig:ligo-cobs-copula}
    \end{subfigure}
    \caption{
        Final copula \glspl{pdf} for \gls{sbevc} and \gls{cobs} after reversal of the survival transformation.
        The data points shown belong to the $\widetilde{\mathcal{D}}_{\text{cop}}$ dataset.
    }
    \label{fig:ligo-copulas}
\end{figure}

Table~\ref{tab:log-likelihood-ligo-cormier} encompasses log-likelihood values of \gls{sbevc} and \gls{cobs} on $\widetilde{\mathcal{D}}_{\text{surv}}$.
The \gls{sbevc} model considered is the original asymmetrical one in \figurename~\ref{fig:ligo-ah}.
Given the apparent similarity between \figurename~\ref{fig:ligo-pickands-cobs} and the instances in \figurename~\ref{fig:simulation-study-gumbel}, we included a Gumbel copula fitted via \gls{mle} in the comparison.
The results confirm the superiority of \gls{sbevc} to \gls{cobs} and the parametric model by a large margin.
The latter is the least fit of the three, just below \gls{cobs}.

\begin{table}
    \centering
    \resizebox{0.5\columnwidth}{!} {
        \begin{tabular}{ccc}
\toprule
$\mathrm{SBEVC}$ & $\mathrm{COBS}$ & $\mathrm{Gumbel}(\theta = 1.87)$ \\
\midrule
           55.47 &           33.39 &                            32.48 \\
\bottomrule
\end{tabular}

    }
    \caption{
        Log-likelihood of different models on the dataset $\widetilde{\mathcal{D}}_{\text{surv}}$.
    }
    \label{tab:log-likelihood-ligo-cormier}
\end{table}

\paragraph{Joint model}

Once fitted both the univariate margin mass model (\glspl{zbs}) and the copula models (\gls{sbevc} and \gls{cobs}), the final joint model immediately followed.
\figurename~\ref{fig:ligo-scatter} plots the original LIGO-Virgo dataset against a random sample generated from each \gls{sbevc} and \gls{cobs} model.
The first and second components are the maximum and the minimum, respectively.
There are ten times more random samples than original data points, for a total of 500.
\figurename~\ref{fig:ligo-scatter-sbevc} and \figurename~\ref{fig:ligo-scatter-cobs} show very similar simulations.
Both capture three main clusters, concentrated in the regions $[0, 20]^2$, $[20, 40] \times [0, 20]$ and $[20, 40]^2$.
It is also worth mentioning that there seems to be a barrier at $\{M_2 = 40\}$; it seems unlikely that giant masses merge.
As pointed out by \figurename~\ref{fig:ligo-copulas}, pictures \figurename~\ref{fig:ligo-scatter-sbevc} and \figurename~\ref{fig:ligo-scatter-cobs} slightly differ in \glspl{bm} with similar masses, being the diagonal just a bit denser in \figurename~\ref{fig:ligo-scatter-cobs}.

\begin{figure}[ht]
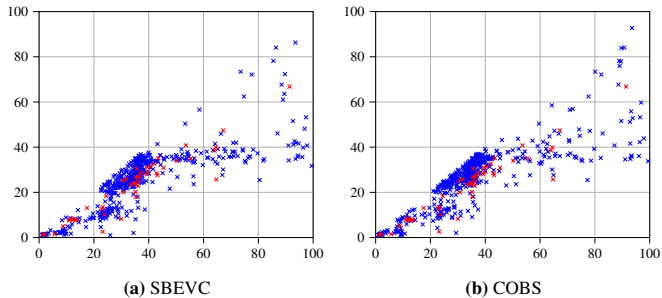

    \center
    \begin{subfigure}{\figurewidth}
        \centering
        \resizebox{\columnwidth}{!}{
            \input{tikz/15-a-ligo-scatter.tex}
        }
        \caption{
            \gls{sbevc}
        }
        \label{fig:ligo-scatter-sbevc}
    \end{subfigure}%
    \begin{subfigure}{\figurewidth}
        \centering
        \resizebox{\columnwidth}{!}{
            \input{tikz/15-b-ligo-cormier-scatter.tex}
        }
        \caption{
            \gls{cobs}
        }
        \label{fig:ligo-scatter-cobs}
    \end{subfigure}
    \caption{
        Original masses (red) against random samples from \gls{sbevc} (blue, on the left) and \gls{cobs} (blue, on the right).
    }
    \label{fig:ligo-scatter}
\end{figure}

    \section{Discussion}

\label{sec:discussion}

\gls{sbevc} is fundamentally different from existing \gls{evc} estimation approaches.
It provides a flexible semiparametric structure, admitting many unconstrained parameters without breaking \gls{pf} assumptions.
Even in the bivariate case, those two feats are difficult to achieve simultaneously~\cite{Vettori2017}.
Moreover, retaining complete control of the parameter space opens up exciting possibilities for statistical modelling and data analysis.
Indeed, \figurename~\ref{fig:random-pickands} and \figurename~\ref{fig:ligo-bayes} represent breakthroughs in \gls{evc} theory.
\figurename~\ref{fig:random-pickands} advances the exploration of the \gls{pf} space with additional smoothing and expressiveness, extending the seminal work by~\citeauthor{Kamnitui2019}~\cite{Kamnitui2019}.
\figurename~\ref{fig:ligo-bayes} shows a Bayesian posterior sample analysis of \glspl{pf}, contributing to solving inferential problems.
Nonparametric approaches, lacking a proper structure, depend on specific samples to build models and can only answer a limited array of inferential questions.

The results obtained in Section~\ref{sec:results} demonstrate the fitting power of \gls{sbevc} on a broad spectrum of \gls{evc} configurations coming from parametric models or even random \glspl{sbevc} like \figurename~\ref{fig:random-pickands}.
Specifically, \gls{sbevc} outperforms a similarly-intended nonparametric approach like \glspl{cobs}~\cite{Cormier2014} on small and medium-sized samples.
This superiority does not lie in the number of parameters, similar in both, but in the more efficient fitting strategy by \gls{sbevc}, especially when data is scarce.
Comparing the top picture in \figurename~\ref{fig:ligo-ah} with \figurename~\ref{fig:ligo-pickands-cobs}, we see that \gls{sbevc} fits a univariate \gls{pdf} via \gls{mle} relying on exact observations, whereas \gls{cobs} attempts a constrained regression on points derived from an empirical copula.
The latter approach will generally be more sensitive to deviations from the \gls{evc} hypothesis, as implied by the fact that some of the fitted points in \figurename~\ref{fig:ligo-pickands-cobs} lie outside the admissible region in \figurename~\ref{fig:pickands}.
Nevertheless, it is remarkable that \gls{sbevc} managed to beat \gls{cobs} in the \gls{rmise} metric, not directly targeted by \gls{sbevc}, which shows the far-reaching capabilities of \gls{mle}.

\glspl{sbevc} represent a vast class of \glspl{evc}.
Two notable copulas fall outside our construction, namely the independence and comonotonic copulas, which correspond to boundary cases of the \gls{pf} geometry.
These limiting cases are usually handled by other means separately and can be approximated in practice through \glspl{sbevc}, as demonstrated in \figurename~\ref{fig:random-pickands}.
Additionally,~\citeauthor{Khoudraji1995}'s device could be applied to refine \glspl{sbevc} in very low correlation settings.
Nonetheless, the construction of \glspl{sbevc} holds the key for fitting even more expressive models by replacing \glspl{zbs} with neural network architectures~\cite{Ling2020}, perhaps at the expense of losing identifiability and a higher risk of overfitting.

Despite all the previous theoretical and practical arguments favouring \gls{sbevc}, the reader may wonder if it is worth the extra execution time and software complexity.
After all, \gls{cobs} can be easily implemented, is already available in \textsf{R}, and provides almost instantaneous results.
What is more, some may even question the practical relevance of complying with \gls{pf} constraints.
Unfortunately, there is no definitive answer to those questions: it depends on the user's goals.
The seeming complexity of \gls{sbevc} is comparable to that of~\citeauthor{HernandezLobato2011}'s proposal.
Theoretical guarantees on the \gls{pf} are nice to have, ensuring the \gls{evc} dependence structure and proper random behaviour in simulation.
In most situations, for exploratory data analysis, a compliant nonparametric method like \gls{cobs} may be the right choice.
\gls{sbevc} may not be a good option if there are tight time constraints.
However, if a more powerful fit is required, data deviates from \gls{evc} assumptions, there is not enough data to confidently apply \gls{cobs}, or one would wish to explore inferential aspects, then \gls{sbevc} might be the better, if not the only one.

    \section{Conclusions}

\label{sec:conclusions}

We have introduced a novel semiparametric approach for estimating bivariate \glspl{evc}.
To our knowledge, it is the first time such an attempt has been made.
\gls{sbevc} allows many parameters while complying with \gls{pf} constraints.
The construction harbours an intriguing potential for Bayesian inference and deep learning.
\glspl{sbevc} represent a vast class of \glspl{evc}, encompassing a broad spectrum of dependence strengths and asymmetries.
Several \glspl{sbevc}' convergence and association properties have been explored.
We have also presented all the algorithms required for effectively and efficiently running the estimation process.
The \gls{ss} shows promising results for \gls{sbevc} in a wide range of sampling configurations.
Specifically, \gls{sbevc} produces significantly lower \gls{rmise} values than \gls{cobs}.
Finally, the case study demonstrates that \gls{sbevc} fits small samples more flexibly than conventional methods.

    \appendix
    \section{Proofs}

\label{sec:proofs}

\begin{proof}[Proof of \thref{prop:A-W-equivalence}]
    Let $W$ be as defined above.
    We will see that $A$ as defined in~\eqref{eq:rotation} is a \gls{pf} with the additional constraint above.

    First, note that $t(0) = 0$, $t(1) = 1$.
    By continuity of $W$, this implies that $\text{Ran}(t) = [0, 1]$.
    Then, for $t(x)$ to be an automorphism of $[0, 1]$, it suffices to see that it is one-to-one.
    Let us suppose that $t(x_1) = t(x_2)$ for some $x_1, x_2 \in [0, 1]$, $x_1 < x_2$.
    Then, $W(x_2) - W(x_1) = x_2 - x_1 > 0$, which leads to a contradiction with $W$ being non-increasing.
    Therefore, $t(x)$ is an automorphism of $[0, 1]$, so $A$
    in~\eqref{eq:rotation} is well-defined as a function of a single variable
    $t \in [0, 1]$.

    Next, letting the support lines $t_+(x) \equiv t(x)$ and $t_{-}(x) \equiv 1 - t(x)$, it is easy to check that $t_{\pm}(x) = \left( 1 \pm x \mp W(x) \right) / 2 $ and, since both $x$ and $W(x)$ are non-negative (otherwise $W$ would not be non-increasing, with \text{Ran}(W) = $[0, 1]$), we may conclude $A(t(x)) \geq \max\{ t_{+}(x), t_{-}(x) \}$.
    Furthermore, $A(t(x)) > 1 - t(x)$ for all $x \in (0, 1] \supset (0, 1/2]$.

    Since $W'(x) \leq 0$ and $W''(x) \geq 0$, it follows that $A''(t) \geq 0$, for all $t \in (0, 1)$, and hence $A$ is convex.
    This finishes the proof that~\eqref{eq:rotation} defines a \gls{pf} such that $A(t) > 1 - t$, for all $t \in (0, 1/2]$.

    Conversely, let $A$ be a \gls{pf} with the latter additional constraint.
    We will similarly show that $W$ as defined in~\eqref{eq:rotation-inverse} is 2-monotone and satisfies $W(0) = 1$ and $W(1) = 0$.

    First, note that $x(0) = 0$ and $x(1) = 0$.
    By continuity of $A$, this implies that $\text{Ran}(x) = [0, 1]$.
    Then, for $x(t)$ to be an automorphism, it suffices to see that $x(t)$ is one-to-one.
    Let us suppose that $x(t_1) = x(t_2)$ for some $t_1, t_2 \in [0, 1]$, $t_1 < t_2$.
    This implies that $[A(t_2) - A(t_1)] / (t_2 - t_1) = -1$ and, since $A$ is convex, we must conclude that $A(t) = 1 - t$ for all $t \in (t_1, t_2]$.
    Clearly, $t_2 \leq 1/2$, because $1 - t < t$ if $t > 1/2$ and, on the other hand, $A(t) \geq \max \{ t, 1 - t \}$.
    Therefore, $(t_1, t_2] \subset (0, 1/2]$, which leads to a contradiction with $A(t) > 1 - t$ over $(0, 1/2]$.
    Hence, $x(t)$ must be one-to-one and, all in all, an automorphism of $[0, 1]$.
    This, in turn, means that $W$ in~\eqref{eq:rotation-inverse} is well-defined as a function of a single variable in $[0, 1]$.

    Next, it is easy to check both $W(0) = 1$ and $W(1) = 0$, bearing in mind that $A(0) = A(1) = 1$.

    Since $A(t) > 1 - t$ for $t \in (0, 1/2]$ and $A$ being convex, we have $A'(t) > -1$ and the denominator in both~\eqref{eq:pickands-first-derivative-inverse} and~\eqref{eq:pickands-second-derivative-inverse} is well-defined.
    Moreover, $A'(t) \leq 1$, otherwise we would have $A(1 - \epsilon) < 1 - \epsilon$ for a sufficiently small $\epsilon$.
    Therefore, $W'(x) \leq x$ for all $x \in (0, 1)$.
    On the other hand, the convexity of $W$ follows directly from $A''(t) \geq 0$.

    Finally, the derivatives~\eqref{eq:pickands-first-derivative} and~\eqref{eq:pickands-second-derivative}, on the one hand, and ~\eqref{eq:pickands-first-derivative-inverse} and~\eqref{eq:pickands-second-derivative-inverse}, on the other, directly follow by differentiating~\eqref{eq:rotation} and~\eqref{eq:rotation-inverse}.
\end{proof}

\begin{proof}[Proof of \thref{prop:fn-converge-imply-wn-converge}]
    It suffices to check that, for all $x \in [0, 1]$,
    \begin{equation*}
        \begin{aligned}
            \lvert W(x) - W_n(x) \rvert
            &=
            \left\lvert
            \int_0^1 \left( 1 - \frac{x}{r} \right)_{+} [f(r) - f_n(r)]
            \ dr
            \right\rvert \\
            &\leq
            \int_0^1 \lvert f(r) - f_n(r) \rvert \ dr
            \,,
        \end{aligned}
    \end{equation*}
    where $(\cdot)_+$ denotes the non-negative part of the argument, and then apply Scheff\'e's theorem~\eqref{eq:scheffe-theorem}.
    Similarly, considering the compact subset $[x_0, 1]$, for some $x_0 > 0$, we have, for all $x \in [x_0, 1]$, $\lvert W'(x) - W_n'(x) \rvert \leq 2 \ d_{\text{TV}}(f, f_n) / x_0$.
\end{proof}

\begin{proof}[Proof of \thref{prop:graph-convergence}]
    It follows from the equivalence between uniform convergence and function graph convergence~\cite{Waterhouse1976} for functions with compact domain and range.
    Since the $W_n$'s uniformly converge to a continuous function $W$, the sequence of the graphs of the $W_n$'s has its limit in the graph of $W$.
    Then, note that the graphs of $A_n$ and $A$ are \textit{affine transformations}~\eqref{eq:rotation} of the graphs of $W_n$ and $W$, respectively.
    This ensures, by continuity, that the graphs of the $A_n$'s tend to that of $A$.
    Finally, graph convergence for the $A_n$'s implies uniform convergence to $A$ itself.

    The result for the first derivatives follows similarly.
    Instead of an affine map, the functions mapping the graph of $W$ to that of $A$ and vice versa are, respectively,
    \begin{multicols}{2}
        \noindent
        \begin{equation*}
            \mathbb{T}(x, w') =
            \left[ t(x), \frac{1 + w'}{1 - w'} \right]
            \,,
        \end{equation*}
        \columnbreak
        \begin{equation*}
            \mathbb{X}(t, a') =
            \left[ x(t), \frac{a' - 1}{a' + 1} \right]
            \,.
        \end{equation*}
    \end{multicols}
    \vspace*{-\multicolsep}
    Both are the inverse of one another because of~\eqref{eq:pickands-first-derivative} and~\eqref{eq:pickands-first-derivative-inverse}.
    Both functions are continuous.
    Hence, they preserve compactness and graph convergence.

    To see that $\{A_n'\}_{n = 1}^{\infty}$ compactly converges to $A'$, consider any compact set $\mathcal{K} = [t_0, 1]$, for $t_0 > 0$.
    Then, consider the sequence of restricted function graphs $\{\mathcal{G}[A_n'|_{\mathcal{K}}]\}_{n = 1}^{\infty}$ and apply $\mathbb{X}$ to every element to obtain another sequence $\{\mathcal{G}[W_n'|_{\mathcal{\mathbb{X}(K)}}]\}_{n = 1}^{\infty}$.
    Now, $\mathbb{X}(\mathcal{K})$ is a compact set, so $W_n'|_{\mathcal{\mathbb{X}(K)}}$ uniformly converges to $W'|_{\mathcal{\mathbb{X}(K)}}$ and, because of ~\cite{Waterhouse1976}, the graph sequence approaches $\mathcal{G}[W'|_{\mathcal{\mathbb{X}(K)}}]$.
    The argument finishes by noting that, since $\mathbb{T}$ is continuous and since $\mathbb{T}(\mathcal{G}[W_n'|_{\mathcal{\mathbb{X}(K)}}]) = \mathcal{G}[A_n'|_{\mathcal{K}}]$ and $\mathbb{T}(\mathcal{G}[W'|_{\mathcal{\mathbb{X}(K)}}]) = \mathcal{G}[A'|_{\mathcal{K}}]$, the graphs of the $A_n'|_{\mathcal{K}}$'s tend to that of $A'|_{\mathcal{K}}$.
\end{proof}

\begin{proof}[Proof of \thref{prop:convergence-tvd-pdfs}]
    First, note that continuity ensures that the limit $p$ is also bounded by the same $K$.
    Denoting $I_q = \int_0^1 e^q$, some easy calculations show that
    \begin{equation*}
        |f(x) - f_n(x)|
        \leq
        \frac
        {e^{p(x)} |I_p - I_{p_n}| + I_p \ |e^{p(x)} - e^{p_n(x)}|}
        {I_p \ I_{p_n}}
        \,.
    \end{equation*}
    Now, the integrals are bounded, namely $e^{-K} \leq I_q \leq e^K$.
    On the other hand, $|e^{p(x)} - e^{p_n(x)}| \leq e^K |p(x) - p_n(x)|$, using the \textit{mean value theorem} and the fact that both functions are bounded by $K$.
    All in all,
    \begin{equation*}
        |f(x) - f_n(x)|
        \leq
        e^{4K} \left( \int_0^1 |p(y) - p_n(y)| \ dy + |p(x) - p_n(x)| \right)
        \,.
    \end{equation*}
    Integrating both sides of the last inequality and using Jensen's inequality, we finally get $d_{\text{TV}}(f, f_n) \leq e^{4K} \lVert p - p_n \rVert_2$.
\end{proof}

\begin{proof}[Proof of \thref{cor:convergence-partial-derivatives}]
    It follows from all the previous convergence results and the form of the partial derivative of an \gls{evc}~\cite{Eschenburg2013, Doyon2013}, where all the terms uniformly converge on compact sets.
\end{proof}

\begin{proof}[Proof of \thref{alg:estimate-W}]
    Let us set aside the normalization by $\bar{W}_0$ for a moment.
    Clearly, $\bar{W}_i''$ is a straight approximation for $W_{\bm{\theta}}''(r_i)$.
    It only remains to check that~\eqref{eq:recurrence-W} and~\eqref{eq:recurrence-W-derivative} provide good approximations for $W_{\bm{\theta}}(r_i)$ and $W_{\bm{\theta}}'(r_i)$, respectively.
    It suffices to see that
    \begin{equation*}
        \bar{W}_i'
        = \sum_{j = i}^n \bar{W}_j' - \bar{W}_{j + 1}'
        =
        - \sum_{j = i}^n Q_j
        \approx
        -\int_{r_i}^{1} \frac{f_{\bm{\theta}}(r)}{r} \ dr
    \end{equation*}
    and, taking into account~\eqref{eq:W-survival},
    \begin{alignat*}{3}
        \bar{W}_i
        &= \sum_{j = i}^n \bar{W}_j - \bar{W}_{j + 1}
        &&=
        \sum_{j = i}^n
        s_j \bar{W}_j' - s_{j + 1} \bar{W}_{j + 1}' + P_j \\
        &=
        s_i \bar{W}_i' +
        \sum_{j = i}^n
        P_j
        &&\approx
        r_i W_{\bm{\theta}}'(r_i) + \int_{r_i}^{1} f_{\bm{\theta}}(r) \ dr
        \,,
    \end{alignat*}
    where we have used that $P_i$ and $Q_i$ are the trapezoidal rule approximations for $\int_{r_i}^{r_{i + 1}} f_{\bm{\theta}}$ and $\int_{r_i}^{r_{i + 1}} f_{\bm{\theta}}(r) / r \ dr$, respectively.
    Also, implicit in the previous argument was the approximation $s_0 = \epsilon \approx 0$, used to avoid infinite values.
    Finally, the last normalization step aims to stabilize the estimation process against numerical errors, enforcing the constraint $W_{\bm{\theta}}(0) = 1$.
\end{proof}

\begin{proof}[Proof of \thref{alg:h-from-a}]
    The rationale of the algorithm is relatively straightforward.
    Equation~\eqref{eq:h-at-grid} mimics~\eqref{eq:h-pdf}, where $A_{\bm{\theta}}$ and its derivatives are evaluated over $t_i$ indirectly through equations~\eqref{eq:rotation},~\eqref{eq:pickands-first-derivative} and~\eqref{eq:pickands-second-derivative}, requiring only the $x_i$'s and approximations of $W_{\bm{\theta}}$ and its derivatives at those points.
    On the other hand, the reader can easily check that $h_{\bm{\theta}}(0) = h_{\bm{\theta}}(1) = 0$, considering all the constraints imposed by \gls{sbevc}: $A_{\bm{\theta}}'(0^+) = -1$, $A_{\bm{\theta}}'(1^-) = 1$ and $f_{\bm{\theta}}(x) > 0$ for all $x \in [0, 1]$, among others.
    The case at the $0$ endpoint is not trivial, but nearly so.
    After simplification, we arrive at
    \begin{equation*}
        h_{\bm{\theta}}(0) = 2 f_{\bm{\theta}}(0)
        \lim_{x \rightarrow 0^+}
        \frac{1 + x - W_{\bm{\theta}}(x)}{x (1 - W_{\bm{\theta}}'(x))^3} \,.
    \end{equation*}
    Repeatedly applying L'H\^opital's rule, we can check that the denominator tends to zero and, eventually, the whole limit also tends to zero.
    Finally, the step involving the integral ensures $\int_0^1 \tilde{h} = 1$, making a $\tilde{h}$ a true \gls{pdf}, which was not automatically granted by the linear interpolation strategy.
\end{proof}

\defbibfilter{references}{
    type=article or
    type=book or
    type=incollection or
    type=inproceedings or
    type=manual or
    type=thesis
}

\defbibfilter{resources}{
    type=software or
    type=online
}

\renewcommand*{\bibfont}{\footnotesize}

\printbibliography[filter=references, title={References}]
\printbibliography[filter=resources, title={Resources}]

\end{document}